\documentclass[review,onefignum,onetabnum]{siamart190516}
\usepackage{lineno,hyperref}

\nolinenumbers
\usepackage{bm}
\usepackage{algorithm}
\usepackage{graphics}
\usepackage{amsfonts}
\usepackage{amssymb}
\usepackage{mathrsfs}

\newtheorem{assumption}{Assumption}[section]

\newtheorem{remark}{Remark}[section]
\def\ban{\begin{eqnarray*}}
\def\ean{\end{eqnarray*}}
\def\bna{\begin{eqnarray}}
\def\ena{\end{eqnarray}}

\title{Distributed order estimation of ARX model under cooperative
excitation condition\thanks{This work was supported by the National Key R\&D Program of China under Grant 2018YFA0703800, National Natural Science Foundation of
China under Grant 11688101,  Natural Science Foundation of Shandong
Province (ZR2020ZD26), and  the Strategic Priority Research Program of Chinese Academy of Sciences under Grant No. XDA27000000.\newline
\hspace*{2em} D. Gan and Z. X. Liu are  with the  Key Laboratory of Systems and Control, Academy of Mathematics and Systems Science, Chinese Academy of Sciences,  and School of Mathematical Sciences, University of Chinese Academy of Sciences, Beijing, P. R. China.
(gandie@amss.ac.cn, lzx@amss.ac.cn)}}

\author{Die Gan, Zhixin Liu\thanks{ Corresponding author.}}

\begin{document}

\maketitle

% REQUIRED
\begin{abstract}
  In this paper, we consider the distributed estimation problem of a linear stochastic system described by an autoregressive model
with exogenous inputs (ARX)
 when both the system orders and parameters are unknown. We design distributed algorithms to estimate the unknown orders and parameters by combining the proposed local information criterion (LIC) with the distributed least squares method. The simultaneous estimation for both the system orders and parameters brings challenges for the theoretical analysis. Some analysis techniques, such as double array martingale limit theory, stochastic Lyapunov functions, and martingale convergence theorems are employed. For the case where the upper bounds of the true orders are available, we introduce a cooperative excitation condition, under which
the strong consistency of the estimation for the orders and parameters is established. Moreover, for the case where the upper
bounds of true orders are unknown, similar distributed algorithm is proposed to estimate both the orders and parameters, and
the corresponding convergence analysis for the proposed algorithm is provided. We remark that our results are obtained without
relying on the independency or stationarity assumptions of regression vectors, and the cooperative excitation conditions can show
that all sensors can cooperate to fulfill the estimation task even though any individual sensor can not.
\end{abstract}

% REQUIRED

% REQUIRED

\section{Introduction}

The statistical models are widely used in almost every field of engineering and science, and how to choose or identify an appropriate statistical models to fit observations is an important issue.  The order estimation of  statistical models is one of the key steps to construct the models. In fact, the investigation of the order estimation has many applications in engineering systems, such as radar \cite{wwc8}, power systems \cite{wwc9}, real seismic  traces \cite{ooorder3} and physiological systems \cite{ooorder4}.

In order  to estimate the order of the statistical models, some criterions are proposed  including  AIC (Akaikes Information Criterion) \cite{wwc4}, BIC (Bayesian Information Criterion) \cite{wwc5},  CIC \cite{wc1} ( the first ``C" emphasizes that the criterion is designed for feedback control systems) and their variants \cite{wwc3}. Based on these information criterions, considerable progresses have been made on the order estimation in time series analysis and adaptive estimation and control (e.g.,\cite{wwc12}-\cite{wwc13}). Some theoretical results are also obtained for the order estimation problem.
For example, Hannan and  Kavalieris in \cite{ooorder5} introduced an algorithm  to estimate the model orders and
system parameters, and the convergence of the algorithm were obtained with stationary input sequence.
Chen and Guo in \cite{wwc3} introduced a modification of the BIC criterion to estimate the order of the multidimensional  ARX  system, where the true orders are assumed to belong to a known finite set. Furthermore, the relevant results for the estimation of the system orders were generalized in \cite{wc5}  to the case where the upper bounds of the true orders are unknown. After that, some development for the order estimation problem are provided (e.g.,  \cite{wwc16}-\cite{wwc17}). Recently, the genetic algorithm \cite{ooorder1} and neural networks \cite{ooorder2} were developed for model order estimation problem with good performance. However, the effectiveness of the proposed algorithms in \cite{ooorder1} and \cite{ooorder2} was verified by some simulation examples without rigorous theoretical analysis.

 %Zhao and Chen in proved the strong consistency of the order estimation of autoregressive moving-average with exogenous input (ARMAX) processes for independently identically distributed input vectors. Ninness inconsidered the Cram$\acute{e}$r-Rao lower bound for model-based spectral density estimation in the case of joint model and model-order estimation with stationary stochastic processes.
%
%However,  many works of existing literature for the order estimation require some sort of independence, stationarity or  ergodicity of the stochastic processes (see e.g.,).

%
%However, they only provided simulations for the effectiveness of the proposed algorithms, and lack of rigorous theoretical analysis.
%We remark that the aforementioned works were carried out in the traditional single sensor case. To the best our knowledge, there is no studies to handle order
%estimation problem in sensor networks.
%
%These results can not be directly applied to systems operating under feedback. The first reference to address the order estimation problem of feedback control systems has been proposed

%The estimation or filtering problems play an important role in stochastic systems because almost all practical systems are subject to
%noisy measurements.
Over the past decade, with the development of communication technology and computer science, wireless sensor networks have attracted increasing attention of researchers due primarily to their practical applications in engineering systems, such as
intelligent transportation and machine health monitoring \cite{wwc31}.
We know that in a sensor network, each sensor can only measure partial
information of the system due to its limited sensing capacity.
In order to estimate the unknown states and system parameters by using data from sensor networks, the centralized and distributed methods are two common schemes, where the latter is gaining increasing
popularity because of scalability, privacy and robustness against node and link
failures. In  distributed algorithms, each sensor only needs to communicate with its neighboring sensors in a certain domain. Some strategies including incremental strategies \cite{wwc21},  consensus
strategies \cite{wwc22}, diffusion strategies \cite{new}, and combination of them \cite{wwc26} are proposed to construct the distributed algorithms. Based on these strategies, the performance analysis of the distributed estimation algorithms are investigated, for example, the consensus-based least mean squares (LMS) (e.g.,\cite{LMS}\cite{LMS2}), the  diffusion stochastic gradient descent algorithm \cite{SG}, the diffusion Kalman filter (e.g., \cite{wwc25}\cite{wwc30}), the diffusion least squares (LS) (e.g., \cite{wwc29}-\cite{oorder3}), the diffusion forgetting factor recursive least squares  \cite{ooorder6}. Most of the corresponding theoretical results are established by requiring the independency, stationarity or Gaussian assumptions for the regression vectors  due to the mathematical difficulty in analyzing
the product of random matrices. However, these requirements are hard to be satisfied since the regression signals may be correlated due to multi-path effect or feedback. In order to avoid using the independency and stationarity conditions of the regressors, some attempts are made for some distributed estimation algorithms. For the time-invariant unknown parameter, Xie and Guo  studied the diffusion LS algorithm and established the convergence result in \cite{wc2}. For time-varying unknown parameter, they investigated the consensus-based and diffusion LMS algorithm, and proposed the corresponding cooperative information condition to guarantee the  stability of the algorithm (e.g.,  \cite{wwc22}\cite{tac} ). For the diffusion Kalman filter algorithm, we introduced the collective random observability condition and provided the stability analysis of distributed Kalman filter algorithm in \cite{wwc30}. We see that the analysis of  all these results are established with known system orders.
How to construct and analyze the distributed algorithms when the system orders are unknown brings challenges for us.

In this paper,  we investigate the distributed estimation problem
 of linear stochastic systems described by an ARX  model with unknown system orders and parameters.
The estimates for the orders of each sensor are obtained by minimizing the proposed LIC, and the estimates for unknown system parameters are derived by the  distributed least squares method where the system orders are replaced by the estimates of orders. The challenges in the theoretical analysis focus on the effect caused by the system noises and the coupled relationship between the estimates of system orders and  parameters. We introduce some mathematical tools including the double array martingale limit theory, martingale convergence theorems, and the stochastic Lyapunov functions to study the convergence of the proposed distributed algorithms. The main contributions of this paper are summarized as follows.
\begin{itemize}
\item
For the case where the true orders have known
upper bounds, we design a distributed algorithm to simultaneously estimate both the system  orders and parameters by minimizing the proposed LIC and using the distributed LS method. A cooperative excitation condition is introduced to reflect the joint effect of multiple sensors: the estimation task can be still completed by the  cooperation of the sensor networks even if any individual sensor can not. Under the cooperative excitation condition, the strong consistency of the estimates for both system orders and parameters is established.

\item
For the case where the upper bounds of true orders are unknown, a similar distributed algorithm is proposed where the growth rate for the upper bounds of the system orders are characterized by a nondecreasing positive function. We employ the double array martingale limit theory to deal with the difficulty arising in analyzing  the cumulative effect of the system noises.   The convergence analysis for system orders and parameters can also be provided.

\item The theoretical results obtained in this paper do not require the
assumptions of the independency and stationarity of the regression signals as used in almost all theoretical analysis of the distributed algorithms, which makes it
possible for applications to the stochastic feedback systems.
\end{itemize}

% The outline is not required, but we show an example here.
The rest of this paper is organized as follows. We introduce some preliminaries including graph theory and the observation model in \cref{section2}.
In \cref{www12}, we establish the convergence results when the upper bounds of the true orders are available. The case where the upper bounds of the true orders are unknown is investigated in \cref{section4}. We present the conclusion of the paper in \cref{section5}.

\section{Problem Formulation}\label{section2}
\subsection{Some Preliminaries}
In this paper, we use $\bm A\in\mathbb{R}^{m\times n}$ to denote an  $m\times n$-dimensional real matrix.
 For a matrix $\bm A$, we use $\lambda_{\max}(\cdot)$ and $\lambda_{\min}(\cdot)$ to denote the largest and smallest eigenvalue of the matrix.
  $\|\bm A\|$ denotes the Euclidean norm, i.e., $\|\bm A\|=(\lambda_{max}(\bm A\bm A^T))^{\frac{1}{2}}$, where the notation $T$ denotes the transpose operator.  We  use  $det(\cdot)$ to denote the determinant of the corresponding matrix.
For a symmetric matrix $\bm A$, if all eigenvalues of $\bm A$ are positive (or nonnegative), then it is a positive definite (semipositive) matrix. Suppose that $\bm A\in\mathbb{R}^{n\times n}$and $\bm B\in\mathbb{R}^{m\times m}$ are two symmetric matrices, and $\bm C$ is an $n\times m$-dimensional matrix. Then by the Rayleigh quotient of the symmetric matrix, we can easily obtain the following inequality,
\bna\label{mineig}
\lambda_{\min}\left(
\begin{matrix}
\bm A & \bm C \\
\bm C^T & \bm B
 \end{matrix}
 \right)\leq \lambda_{\min}(\bm A).
\ena
The matrix inversion formula is used in our analysis, and we list it here.
\begin{lemma}\cite{wc3} \label{wl1}
For any matrices $\bm A$, $\bm B$,  $\bm C$ and  $\bm D$ with suitable dimensions, the following formula
\ban
(\bm A+\bm B\bm D\bm C)^{-1}=\bm A^{-1}-\bm A^{-1}\bm B(\bm D^{-1}+\bm C\bm A^{-1}\bm B)^{-1}\bm C \bm A^{-1}
\ean
holds, provided that the relevant matrices are invertible.
\end{lemma}

If all elements of a matrix $\bm A=\{a_{ij}\}\in\mathbb{R}^{n\times n}$ are nonnegative, then it is a nonnegative matrix, and furthermore if  $\sum^n_{j=1}a_{ij}=1$ holds for all $i\in\{1,\cdots,n\}$, then it is called a stochastic matrix.

Let $\{\bm A_k\}$  be a matrix sequence and $\{ b_k\}$  be a positive scalar
sequence. Then by $\bm A_k = O(b_k)$ we mean that there exists a constant $C > 0$ such that
$\| \bm A_k\|  \leq  C b_k, ~\forall k \geq  0$, and by $\bm A_k=o(b_k)$ we mean that $\lim_{k\rightarrow \infty}\|\bm A_k\|/b_k=0$.

In this paper, our purpose is to estimate both system orders and parameters in a distributed way and establish the corresponding convergence results. We use an undirected  graph ${\mathcal{G}}=(\mathcal{V},\mathcal{E})$  to describe the relationship between sensors where $\mathcal{V}$ is the set of sensors and $\mathcal{E}$ is the edge set. The adjacency matrix $\mathcal{A}=\{a_{ij}\}\in\mathbb{R}^{n\times n}$ is introduced to reflect the weight of the corresponding edge. The elements of $\mathcal{A}$ satisfy:  $a_{ij}>0$ if $(i, j)\in \mathcal{E}$ and $a_{ij}=0$ otherwise. The set of the neighbors of sensor $i$ is denoted as $N_i=\{j\in \mathcal{V}|(i,j)\in\mathcal{E} \}$. A path of length $\ell$ is a sequence of $\ell+1$ sensors such that the subsequent senors are connected. The graph {$\mathcal{G}$} is called connected if for any two sensors $i$ and $j$, there is a path connecting them. The diameter {$D_{\mathcal{G}}$} of the graph {$\mathcal{G}$} is defined as the maximum length of the path between any two sensors. For simplicity of analysis, the convergence of the estimates in this paper is considered under the condition that the weighted adjacency matrix $\mathcal{A}$ is symmetric and stochastic. Thus, it is obvious that $\mathcal{A}$ is doubly stochastic.

\subsection{Observation model}

We consider a network composed of $n$ sensors. At each time instant $t ~(t=0,1,2,\cdots)$, the input signal $u_{t,i}\in \mathbb{R}$ and the output signal $y_{t,i}\in \mathbb{R}$  of  sensor  $i\in\{1,...,n\}$ are assumed to obey the following linear stochastic ARX model,
\begin{gather}
y_{t+1,i}=b_1y_{t,i}+\cdots+b_{p_0}y_{t+1-{p_0},i}+c_1u_{t,i}+\cdots+c_{q_0}u_{t+1-q_0,i}+w_{t+1,i},\label{www14}\\
y_{t,i}=0, ~~ u_{t,i}=0, ~~\hbox{for}~ ~t\leq0,\nonumber
\end{gather}
where $\{w_{t,i}\}$ is a noise process, $p_0$, $q_0$ are unknown true orders $(b_{p_0}\neq0, c_{q_0}\neq0)$ and $b_1, \cdots, b_{p_0}$, $ c_1,\cdots, c_{q_0} $ are unknown parameters.

Denote the unknown parameter vector $\bm \theta(p,q)$ and the corresponding regression vector $\bm \varphi_{t,i}(p,q)$ as
\bna
&& \bm \theta(p,q)=[b_1,\cdots,b_{p}, c_1,\cdots, c_{q}]^T,\\ &&
\bm \varphi_{t,i}(p,q)=[y_{t,i},\cdots,y_{t+1-{p},i},u_{t,i},\cdots,u_{t+1-q,i}]^T.
\ena
If $p>p_0$, then $b_j=0$ for $p_0<j\leq p$ ; and if $q>q_0$, then $c_m=0$ for $q_0<m\leq q$.
The regression model (\ref{www14}) can be rewritten as
\bna
y_{t+1,i}&=&\bm\theta^T(p,q)\bm \varphi_{t,i}(p,q)+w_{t+1,i}\ \ (\hbox{for}\  \ p\geq p_0\ \  \hbox{and}\ \ q\geq q_0)\label{dynamics} \\&=&\bm\theta^T(p_0,q_0)\bm \varphi_{t,i}(p_0,q_0)+w_{t+1,i}\label{2.6}.
\ena

The purpose of this paper is to design the distributed algorithm for each sensor by using the local information from its neighbors to estimate both  the system orders $p_0, q_0$ and the parameter vector $\bm\theta(p_0,q_0)$ . We know that for the case where the system orders $p_0, q_0$ are known, the distributed LS algorithm
is one of the most basic algorithms to estimate the unknown parameters, and it has wide applications because of its fast convergence rate, e.g., in the area of cloud technologies (e.g., \cite{order1}). The details of the distributed LS algorithm can be found in the following \Cref{algorithm1} (see \cite{wc2}).

\begin{algorithm}[!ht]
\caption{Distributed LS Algorithm}\label{algorithm1}
For any given $i\in\{1,\cdots,n\}$ and given system order $(p,q)$, begin with an initial estimate $\bm \theta_{0,i}(p,q)$, and an
initial positive definite matrix $\bm P_{0,i}(p,q)$. The distributed LS algorithm is recursively defined at time instant $t\geq0$ as follows.

1: Adaptation.
\bna
\bar{\bm \theta}_{t+1,i}(p,q)&=&\bm \theta_{t,i}(p,q)+d_{t,i}(p,q)\bm P_{t,i}(p,q)\bm \varphi_{t,i}(p,q)\nonumber\\
&&\cdot(y_{t+1,i}-\bm \varphi^T_{t,i}(p,q)\bm \theta_{t,i}(p,q)),\quad \label{ws1}\\
\bar{\bm P}_{t+1,i}(p,q)&=&\bm P_{t,i}(p,q)-d_{t,i}(p,q)\bm P_{t,i}(p,q)\bm \varphi_{t,i}(p,q)\bm \varphi^T_{t,i}(p,q)\bm P_{t,i}(p,q),\label{ws2}\\
d_{t,i}(p,q)&=&[1+\bm \varphi^T_{t,i}(p,q)\bm P_{t,i}(p,q)\bm \varphi_{t,i}(p,q)]^{-1}\label{ws3},
\ena

2: Diffusion.
\bna
&\bm P^{-1}_{t+1,i}(p,q)&=\sum_{j\in N_i}a_{ij}\bar{\bm P}^{-1}_{t+1,j}(p,q),\label{ws4}\\
&\bm \theta_{t+1,i}(p,q)&=\bm P_{t+1,i}(p,q)\sum_{j\in N_i}a_{ij}\bar{\bm P}^{-1}_{t+1,j}(p,q)\bar{\bm \theta}_{t+1,j}(p,q).\label{ws5}
\ena
\end{algorithm}

In this section, for any given $(p, q)$, the estimation error between the true parameter and the estimate obtained by \Cref{algorithm1} is denoted as $\widetilde{\bm \theta}_{t,i}(p,q)$,

\bna
\widetilde{\bm \theta}_{t,i}(p,q)=[b_1-b^i_{1,t},\cdots,b_{p}-b^i_{p,t}, c_1-c^i_{1,t},\cdots, c_{q}-c^i_{q,t}]^T,\label{error}
\ena
where $\{b^i_{j,t}\}^{p}_{j=1}$ and  $\{c^i_{r,t}\}^q_{r=1}$ are denoted as the estimates of the corresponding
components of $\bm \theta_{t,i}(p,q)$ obtained by \Cref{algorithm1}.

We have the following result on the estimation error $\widetilde{\bm \theta}_{t,i}(p,q)$, which will be helpful for the subsequent theoretical analysis.

\begin{lemma}\label{wl2}
For $p\geq p_0$ and $q\geq q_0$, the following equation holds,
\bna\label{we1}
~~~~~~~\bm P^{-1}_{t+1,i}(p,q)\widetilde{\bm \theta}_{t+1,i}(p,q)=\sum_{j\in N_i}a_{ij}\bm P^{-1}_{t,j}(p,q)\widetilde{\bm \theta}_{t,j}(p,q)-\sum_{j\in N_i}a_{ij}\bm \varphi_{t,j}(p,q)w_{t+1,j}.
\ena
\end{lemma}

\begin{proof}
For simplicity of expression, we use $d_{t,i}$, $\bm\varphi_{t,i}$, $\bm P_{t,i}$, $\bar{\bm P}_{t+1,i} $, $ \bar{\bm \theta}_{t+1,i}$, $\widetilde{\bm \theta}_{t,i}$ and  $\bm \theta_{t+1,i}$ to denote $d_{t,i}(p, q)$, $\bm\varphi_{t,i}(p, q)$, $\bm P_{t,i}(p, q)$, $\bar{\bm P}_{t+1,i}(p, q)$, $ \bar{\bm \theta}_{t+1,i}(p, q)$, $\widetilde{\bm \theta}_{t,i}(p,q)$ and  $\bm \theta_{t+1,i}(p, q)$. By (\ref{ws3}), we have
\bna
d_{t,i}=1-d_{t,i}\bm \varphi^T_{t,i}\bm P_{t,i}\bm \varphi_{t,i}.\label{wnew2}
\ena
Combining this with (\ref{ws1}) and (\ref{ws2}), we have
\ban
\bar{\bm \theta}_{t+1,i}&=&(\bm I-d_{t,i}\bm P_{t,i}\bm \varphi_{t,i}\bm \varphi^T_{t,i})\bm \theta_{t,i}
+d_{t,i}\bm P_{t,i}\bm \varphi_{t,i}y_{t+1,i}\\
&=&(\bm I-d_{t,i}\bm P_{t,i}\bm \varphi_{t,i}\bm \varphi^T_{t,i})\bm \theta_{t,i}
+\bm P_{t,i}\bm \varphi_{t,i}(1-d_{t,i}\bm \varphi^T_{t,i}\bm P_{t,i}\bm \varphi_{t,i})y_{t+1,i}\\
&=&(\bm P_{t,i}-d_{t,i}\bm P_{t,i}\bm \varphi_{t,i}\bm \varphi^T_{t,i}\bm P_{t,i})\bm P^{-1}_{t,i}\bm \theta_{t,i}
+(\bm P_{t,i}-d_{t,i}\bm P_{t,i}\bm \varphi_{t,i}\bm \varphi^T_{t,i}\bm P_{t,i})\bm \varphi_{t,i}y_{t+1,i}\\
&=&\bar{\bm P}_{t+1,i}\bm P^{-1}_{t,i}\bm \theta_{t,i}+\bar{\bm P}_{t+1,i}\bm \varphi_{t,i}y_{t+1,i}.
\ean
Hence we have
\ban
{\bar{\bm P}}^{-1}_{t+1,i}\bar{\bm \theta}_{t+1,i}=\bm P^{-1}_{t,i}\bm \theta_{t,i}+\bm \varphi_{t,i}y_{t+1,i}.
\ean
Substituting this equation into (\ref{ws5}) yields
\bna
\bm P^{-1}_{t+1,i}\bm \theta_{t+1,i}=\sum_{j\in N_i}a_{ij}(\bm P^{-1}_{t,j}\bm \theta_{t,j}+\bm \varphi_{t,j}y_{t+1,j})\label{w6}.
\ena
By (\ref{ws2}) and \Cref{wl1}, we have
\bna
{\bar{\bm P}}^{-1}_{t+1,i}={{\bm P}}^{-1}_{t,i}+\bm\varphi_{t,i}\bm\varphi^T_{t,i}.\label{w8}
\ena
Hence by (\ref{ws4}), (\ref{w6}) and (\ref{w8}), we have
\ban
&&\bm P^{-1}_{t+1,i}\widetilde{\bm \theta}_{t+1,i}=\bm P^{-1}_{t+1,i}\bm\theta- \bm P^{-1}_{t+1,i}{\bm \theta}_{t+1,i}\\
&=&\sum_{j\in N_i}a_{ij}{\bar{\bm P}}^{-1}_{t+1,j}\bm\theta-  \sum_{j\in N_i}a_{ij}(\bm P^{-1}_{t,j}\bm \theta_{t,j}+\bm \varphi_{t,j}y_{t+1,j})\\
&=&\sum_{j\in N_i}a_{ij}({{\bm P}}^{-1}_{t,j}+\bm\varphi_{t,j}\bm\varphi^T_{t,j})\bm\theta-  \sum_{j\in N_i}a_{ij}(\bm P^{-1}_{t,j}\bm \theta_{t,j}+\bm \varphi_{t,j}\bm\varphi^T_{t,j}\bm\theta+\bm \varphi_{t,j}w_{t+1,j})\\
&=&\sum_{j\in N_i}a_{ij}\bm P^{-1}_{t,j}\widetilde{\bm \theta}_{t,j}-\sum_{j\in N_i}a_{ij}\bm \varphi_{t,j}w_{t+1,j},
\ean
which completes the proof of the lemma.
\end{proof}

For the case where the system orders $p_0, q_0$ are known, Xie and Guo in \cite{wc2} proved that the distributed LS algorithm can converge to the true parameters almost surely (a.s.) under a cooperative excitation condition. However, when the system orders $p_0, q_0$ are unknown, the estimation for both the system orders and the parameters makes the design and analysis of the distributed algorithms quite complicated. We will deal with such a problem in the following two sections.

\section{Case I: The upper bounds of true orders are known}\label{www12}

In this section, we will first design the distributed algorithm to estimate both the system orders $(p_0, q_0)$  and the parameter vector $\bm \theta(p_0,q_0)$  for the case where the system orders have known upper bounds, i.e.,
\ban
(p_0, q_0)\in M\triangleq\{(p,q), 0\leq p\leq p^*, 0\leq q\leq q^*\},
\ean
where $p^*$ and $q^*$  are known upper bounds of the system orders.

For convenience of analysis, we introduce some  notations and assumptions,
 \ban &&\bm d_t(p,q)=diag\{d_{t,1}(p,q),...d_{t,n}(p,q)\},\nonumber\\
  &&\bm \Phi_t(p,q)= diag\{\bm\varphi_{t,1}(p,q),...\bm\varphi_{t,n}(p,q)\},\nonumber\\
  &&\bm W_t=[w_{t,1},...w_{t,n}]^{T},\nonumber\\
&&\bm P_t(p,q)= diag\{\bm P_{t,1}(p,q),...,\bm P_{t,n}(p,q)\},\nonumber\\
  &&\bar{\bm P}_t(p,q)=diag\{\bar{\bm P}_{t,1}(p,q),...,\bar{\bm P}_{t,n}(p,q)\},\nonumber\\
  &&\widetilde{\bm \Theta}_t(p,q)=col\{\widetilde{\bm \theta}_{t,1}(p,q),...,\widetilde{\bm \theta}_{t,n}(p,q)\},\ean
where $col(\cdot,\cdots,\cdot)$ denotes a vector stacked by the specified vectors, and $diag(\cdot,\cdots,\cdot)$ denotes a block matrix formed in a diagonal manner of the corresponding vectors or matrices.

In order to propose and further analyze the distributed algorithm used to estimate both the system order and the parameters, we introduce some assumptions on the network topology, and the observation noise and the regression vectors.

\begin{assumption}\label{a1}
The communication graph $\mathcal{G}$ is connected.
\end{assumption}

\begin{remark}\label{remark31} Denote $\mathcal{A}^l\triangleq[a_{ij}^{(l)}]$ { with $\mathcal{A}$ being the weighted adjacency matrix of graph $\mathcal{G}$,  i.e., $a_{ij}^{(l)}$
is the $(i,j)$-th entry of
the matrix $\mathcal{A}^l, l\geq 1$ and  $a_{ij}^{(1)}=a_{ij}$.}   Under \Cref{a1}, we can easily obtain that $\mathcal{A}^l$ is a positive matrix for $l\geq D_{\mathcal{G}}$ by the theory of product of stochastic matrices, which means that  $a_{ij}^{(l)}>0$ for any $i$ and $j$. \end{remark}

\begin{assumption}\label{a2}
For any $i\in\{1,\cdots,n\}$, the noise sequence $\{w_{t,i},\mathscr{F}_t\}$ is a martingale difference sequence where
$\mathscr{F}_t$  is a sequence of nondecreasing $\sigma$-algebras generated by $\{y_{k,i}, u_{k,i}, k\leq t, i=1,2,,\cdots, n\}$, and
 there exists a constant $\beta > 2$
such that
\ban
\sup_{t\geq 0}E[|w_{t+1,i}|^{\beta}|\mathscr{F}_t]<\infty, {\rm a.s.}
\ean
where $E[\cdot|\cdot]$ denotes  the conditional expectation operator.
\end{assumption}

\begin{assumption}\label{a5}{\rm (Cooperative Excitation Condition I).}
There exists a sequence $\{a_t\}$ of  positive real numbers  satisfying ${a_t}\xrightarrow[{t\rightarrow\infty}] {}$ $\infty$ and
\bna
\frac{\log {r_t(p^*,q^*)}}{a_t}\xrightarrow[{t\rightarrow\infty}] {} 0, \ \ \ \ \frac{ {a_t}}{\lambda^{p,q}_{\min}(t)}\xrightarrow[{t\rightarrow\infty}] {} 0,\ \ \hbox{for}\ \ (p,q)\in M^* ~~~{\rm a.s.}
\ena
where $M^*=\{(p_0,q^*),(p^*,q_0)\}$, $r_t(p,q)=\lambda_{\max}\{\bm P^{-1}_0(p,q)\}+\sum^n_{i=1}\sum^t_{k=0}\|\bm\varphi_{k,i}(p,q)\|^2$ and
\ban
\lambda^{p,q}_{\min}(t)= \lambda_{\min}\left\{\sum^n_{j=1}\bm P^{-1}_{0,j}(p,q)+\sum^n_{j=1}\sum ^{t-D_{\mathcal{G}+1}}_{k=0}\bm \varphi_{k,j}(p,q){\bm \varphi}^T_{k,j}(p,q)\right\}.
\ean
\end{assumption}

\begin{remark} We give an explanation for the choice of $\{a_t\}$ in \Cref{a5} for two typical cases: (I) If the regression vectors $\bm\varphi_{k,i}(p^*,q^*)$ are bounded for any $i\in\{1,...,n\}$, and satisfy the ergodicity property, i.e., there exists a matrix $\bm H_i$ such that
$
\frac{1}{t}\sum^t_{k=1}\bm\varphi_{k,i}(p^*,q^*){\bm\varphi^T_{k,i}(p^*,q^*)}\xrightarrow[{t\rightarrow\infty}] {} \bm H_i
$
with $\sum^n_{i=1}\bm H_i$ being positive definite (see e.g., \cite{c7}), then $a_t$ can be taken as $a_t=t^{\rho}$, $0<\rho<1$;
(II) If there exist three positive constants $s_1, s_2$ and $s_3$ (they may depend on the sample  $\omega$) such that
\ban
\sum^n_{i=1}\sum^{t}_{k=0}(\|y_{k,i}\|^2+\|u_{k,i}\|^2)=O(t^{s_1}), ~~ {\rm a.s.}\nonumber\\
\lambda^{p,q}_{\min}(t)\geq s_2(\log t)^{1+s_3}, ~\hbox{for} ~(p,q)\in M^*,~~{\rm a.s.},
\ean
then \Cref{a5} can be also satisfied by taking $a_t=(\log t)\log\log t$.
\end{remark}

\begin{remark}
For the case where there is only one sensor $(n=1)$, Guo et al. in \cite{wc1} investigated the strong consistency of the order estimate under the following
conditions,
\begin{equation}
\begin{split}
&\frac{\log ({\sum^{t}_{k=0}\|\bm \varphi_{k,1}(p^*,q^*)\|^2}+1)}{a_t}\xrightarrow[{t\rightarrow\infty}] {} 0, ~~{\rm a.s.}\\
&\frac{ {a_t}}{\lambda_{\min}(\sum^t_{k=0}\bm \varphi_{k,1}(p,q)\bm \varphi^T_{k,1}(p,q)+\gamma\bm I)}\xrightarrow[{t\rightarrow\infty}] {} 0,  ~~\hbox{for} ~(p,q)\in M^*,~~{\rm a.s.},\label{order1}
\end{split}
\end{equation}
where $\gamma$ is a positive constant, and  $\{a_t, t\geq 1\} $  is a sequence of positive numbers.  \Cref{a5}
can be considered as an extension of (\ref{order1}) to the case of multiple sensors.
\end{remark}

\begin{remark}\label{joint} Cooperative Excitation Condition I (i.e., \cref{a5}) can reflect the joint effect of multiple sensors to some extent: all sensors may cooperatively estimate the unknown orders and parameters  under \cref{a5}(see \cref{t1} and \cref{knownparameter}), even though any individual sensor can not fulfill the estimation task since the single sensor may be lack of adequate excitation to satisfy the condition (\ref{order1}).
\end{remark}

In the following, we propose an algorithm to estimate the system orders $p_0$ and $q_0$ in a distributed way. For this propose, we introduce a local information criterion  $L_{t,i}(p,q)$ for the sensor $i$ at the time instant $t\geq 0$,
\bna\label{criterion}
L_{t,i}(p,q)=\sigma_{t,i}(p,q,\bm\theta_{t,i}(p,q))+(p+q)a_t,
\ena
where  $\sigma_{0,i}(p,q,\bm\beta(p,q))=0$, and $\sigma_{t,i}(p,q,\bm\beta(p,q))$ is recursively defined for $t>0$ as follows,
\bna\label{3}
~~~~~\sigma_{t,i}(p,q,\bm\beta(p,q))=\sum_{j\in N_i}a_{ij}\big(\sigma_{t-1,j}(p,q,\bm\beta(p,q))+[y_{t,j}-{\bm\beta}^T(p,q)\bm\varphi_{t-1,j}(p,q)]^2\big).
\ena
By $\sigma_{0,i}(p,q,\bm\beta(p,q))=0$, (\ref{3}) is equivalent to the following equation,
\bna
\sigma_{t,i}(p,q,\bm\beta(p,q))=\sum^n_{j=1}\sum^{t-1}_{k=0}a^{(t-k)}_{ij}[y_{k+1,j}-{\bm\beta}^T(p,q)\bm\varphi_{k,j}(p,q)]^2.\label{w4}
\ena

When the upper bounds of orders are known, the distributed algorithm to estimate the system orders and parameters is put forward by minimizing LIC ( i.e., $L_{t,i}(p,q)$) and using \Cref{algorithm1}. It is clear that in (\ref{criterion}), the first term is used to minimize the error between the observation signals and the prediction, while the
penalty term ``$(p+q)a_t$"  is introduced to avoid overfitting. The details of the algorithm can be found in \Cref{algorithm2}.

\begin{algorithm}[!ht]
\caption{ }\label{algorithm2}
 For any given $i\in\{1,\cdots,n\}$,  the distributed estimation  of the system orders and parameters can be obtained at the time instant $t\geq1$ as follows.

$\mathbf{Step\ 1:}$ For any $(p,q)\in M$, based on $\{\bm \varphi_{k,j}(p,q), y_{k+1,j}\}^{t-1}_{k=1}$, $ j\in N_i,$
 the estimate  $\bm\theta_{t,i}(p,q)$ can be obtained by using \Cref{algorithm1}.

$\mathbf{Step\ 2:}$ (Order Estimation) With the estimates $\{\bm\theta_{t,i}(p,q)\}_{(p,q)\in M}$ obtained by Step 1, the estimates  $(p_{t,i},q_{t,i})$ of  system orders are
given by minimizing $L_{t,i}(p,q)$, i.e.,
\bna\label{1}
(p_{t,i},q_{t,i})={\arg\min}_{(p,q)\in M}L_{t,i}(p,q).
\ena

$\mathbf{Step\ 3:}$ (Parameter Estimation)
The  estimate $\bm\theta_{t,i}(p_{t,i},q_{t,i})$ for the unknown parameter $\bm\theta(p_0,q_0)$ can be obtained by using \Cref{algorithm1}, where the orders $(p, q)$ are replaced by the estimates $(p_{t,i},q_{t,i})$ obtained in  Step 2.

Repeating the above steps, we obtain the order estimates $p_{t,i},q_{t,i}$  and parameter estimates $\bm\theta_{t,i}(p_{t,i},q_{t,i})$ for $t\geq 0$ and $i=1,2,\cdots,n$.
\end{algorithm}

In the following, we will analyze the convergence of the estimation for system orders and  parameters obtained in \Cref{algorithm2}. To this end, we first introduce some preliminary lemmas.

\begin{lemma}{ \rm\cite{wc2}}\label{wla}
In \Cref{algorithm1}, for any fixed $p, q$ and $t\geq 1$, we have
\ban
\lambda_{\max}\{\bm d_t(p,q)\bm \Phi^T_t(p,q)\bm P_t(p,q)\bm \Phi_t(p,q)\}\leq \frac{det (\bm P^{-1}_{t+1}(p,q))-det(\bm P^{-1}_{t}(p,q))}{det(\bm P^{-1}_{t+1}(p,q))}\leq1.
\ean
\end{lemma}

\begin{lemma}{\rm\cite{wc2}}\label{l3}
Under Assumptions \ref{a1} and  \ref{a2},   we have for $p\geq p_0$ and $q\geq q_0$,
\ban
\sum^n_{i=1}\widetilde{\bm\theta}^T_{t,i}(p,q){{\bm P}}^{-1}_{t,i}(p,q) \widetilde{\bm\theta}_{t,i}(p,q)=O(\log r_t(p,q)),
\ean
where $r_t(p,q)$ is defined in \Cref{a5}.
\end{lemma}

How to deal with the effect of the noises is a crucial step for the convergence analysis of \cref{algorithm2}, and the following lemma provides an upper bound for
the cumulative summation of the noises.
\begin{lemma}\label{l5}
Under Assumptions \ref{a1} and \ref{a2},  we have for any fixed $p, q$,
\ban
\bm S^T_{t+1,i}(p,q)\bm P_{t+1,i}(p,q)\bm S_{t+1,i}(p,q)= O(\log r_{t}(p,q)),
\ean
where  $\bm S_{t+1,i}(p,q)=\sum^n_{j=1}\sum^{t}_{k=0}a^{(t+1-k)}_{ij}\bm\varphi_{k,j}(p,q)w_{k+1,j}$,  and $a^{(t+1-k)}_{ij}$ is the $i$-th row, $j$-th column entry of the weight matrix $\mathcal{A}^{t+1-k}$.
\end{lemma}

\begin{proof}
For the convenience of expression, we use $\bm S_{t,i}, \bm P_{k}, \bm \Phi_k$ and $\bm d_k$ to denote $\bm S_{t,i}(p,q)$, $ \bm P_{k}(p,q )$, $\bm \Phi_k(p,q)$ and $\bm d_k(p,q)$.

Set $\bm S_0=0$ and $\bm S_t=col\{\bm S_{t,1},...,\bm S_{t,n}\}.$
Then  we have
\ban
\bm S_{k+1}=\sum^{k}_{l=0}\mathscr{A}^{k+1-l}\bm \Phi_l\bm W_{l+1}=\mathscr{A}(\bm S_k+\bm \Phi_k\bm W_{k+1}).
\ean
By (\ref{ws2}) and Lemma 4.2 in \cite{wc2}, we have
\bna & &
\bm S^T_{k+1} \bm P_{k+1} \bm S_{k+1}=(\bm S^T_k+\bm W^T_{k+1}\bm \Phi^T_k) \mathscr{A}^T\bm P_{k+1}\mathscr{A}(\bm S_k+\bm \Phi_k\bm W_{k+1})\nonumber\\
&\leq & (\bm S^T_k+\bm W^T_{k+1}\bm \Phi^T_k)\bar{\bm P}_{k+1}(\bm S_k+\bm \Phi_k\bm W_{k+1})\nonumber\\
&=&  (\bm S^T_k+\bm W^T_{k+1}\bm \Phi^T_k)({\bm P}_{k}-\bm P_k\bm \Phi_k\bm d_k\bm \Phi^T_k \bm P_k)(\bm S_k+\bm \Phi_k\bm W_{k+1})\nonumber\\
&=& \bm S^T_{k} \bm P_{k} \bm S_{k}+2\bm W^T_{k+1}\bm \Phi^T_k\bm P_k\bm S_k+\bm W^T_{k+1}\bm \Phi^T_k {\bm P}_{k}\bm \Phi_k\bm W_{k+1}-\bm S^T_k\bm P_k\bm \Phi_k\bm d_k\bm \Phi^T_k \bm P_k\bm S_k\nonumber\\
&&-2\bm S^T_k\bm P_k\bm \Phi_k\bm d_k\bm \Phi^T_k \bm P_k\bm \Phi_k \bm W_{k+1}-\bm W^T_{k+1}\bm \Phi^T_k\bm P_k\bm \Phi_k\bm d_k\bm \Phi^T_k \bm P_k\bm \Phi_k\bm W_{k+1}.\label{wnew1}
\ena
Moreover, by the definition of $\bm d_k$ and (\ref{wnew2}), we have
\bna
\bm d_k\bm \Phi^T_k \bm P_k\bm \Phi_k=\bm I-\bm d_k.\label{wnew3}
\ena
By (\ref{ws2}) and (\ref{wnew3}), we derive that
\bna
\bar{\bm P}_{k+1}\bm \Phi_k&=&{\bm P}_{k}\bm \Phi_k-\bm P_k\bm \Phi_k\bm d_k\bm \Phi^T_k \bm P_k\bm \Phi_k\nonumber\\
&=&{\bm P}_{k}\bm \Phi_k-\bm P_k\bm \Phi_k(\bm I-\bm d_k)={\bm P}_{k}\bm \Phi_k\bm d_k.\label{wnew4}
\ena
Substituting (\ref{wnew3}) into (\ref{wnew1}),  we have by (\ref{wnew4})
\ban
\bm S^T_{k+1} \bm P_{k+1} \bm S_{k+1}&\leq&\bm S^T_{k} \bm P_{k} \bm S_{k}+2\bm S^T_k\bm P_k\bm \Phi_k\bm d_k\bm W_{k+1}
-\bm S^T_k\bm P_k\bm \Phi_k\bm d_k\bm \Phi^T_k \bm P_k\bm S_k\\
&&+\bm W^T_{k+1}\bm \Phi^T_k\bm P_k\bm \Phi_k\bm d_k\bm W_{k+1}\\
&=&\bm S^T_{k} \bm P_{k} \bm S_{k}+2\bm S^T_k\bar{\bm P}_{k+1}\bm \Phi_k\bm W_{k+1}
-\bm S^T_k\bar{\bm P}_{k+1}\bm \Phi_k\bm d^{-1}_k\bm \Phi^T_k \bar{\bm P}_{k+1}\bm S_k\\
&&+\bm W^T_{k+1}\bm \Phi^T_k\bm P_k\bm \Phi_k\bm d_k\bm W_{k+1}\\
&\leq&\bm S^T_{k} \bm P_{k} \bm S_{k}+2\bm S^T_k\bar{\bm P}_{k+1}\bm \Phi_k\bm W_{k+1}
-\bm S^T_k\bar{\bm P}_{k+1}\bm \Phi_k\bm \Phi^T_k \bar{\bm P}_{k+1}\bm S_k\\
&&+\bm W^T_{k+1}\bm \Phi^T_k\bm P_k\bm \Phi_k\bm d_k\bm W_{k+1}.
\ean
By the summation of both sides of the above inequality, we have
\bna
&&\bm S_{t+1}\bm P_{t+1}\bm S_{t+1}+\sum^{t}_{k=0}\|\bm S^T_k\bar{\bm P}_{k+1}\bm \Phi_k\|^2\nonumber\\
&\leq&2\sum^{t}_{k=0}\bm S^T_k\bar{\bm P}_{k+1}\bm \Phi_k\bm W_{k+1}+\sum^{t}_{k=0}\bm W^T_{k+1}\bm \Phi^T_k\bm P_k\bm \Phi_k\bm d_k\bm W_{k+1}.\label{wnew5}
\ena
Next, we estimate the two terms on the right hand side of (\ref{wnew5}) separately. By \Cref{a2} and the martingale estimation theorem (see,e.g.,\cite{wc4}), we can get the following inequality,
\bna
\sum^{t}_{k=0}\bm S^T_k\bar{\bm P}_{k+1}\bm \Phi_k\bm W_{k+1}=O(1)+o\left(\sum^{t}_{k=0}\|\bm S^T_k\bar{\bm P}_{k+1}\bm \Phi_k\|^2\right).\label{wnew6}
\ena
Then  by the proof of Lemma 4.4 in \cite{wc2}, we obtain
\bna
\sum^{t}_{k=0}\bm W^T_{k+1}\bm \Phi^T_k\bm P_k\bm \Phi_k\bm d_k\bm W_{k+1}=\sum^{t}_{k=0}\bm W^T_{k+1}\bm d_k\bm \Phi^T_k\bm P_k\bm \Phi_k\bm W_{k+1}=O(\log r_{t}).\label{new7}
\ena
Substituting (\ref{wnew6}) and (\ref{new7}) into  (\ref{wnew5}) yields
\ban
\bm S_{t+1}\bm P_{t+1}\bm S_{t+1}=O(\log r_{t}),
\ean
which completes the proof of the lemma.
\end{proof}

\begin{remark}\label{rrr1}
If \Cref{a2} is relaxed to the following weaker noise condition
\bna
\sup_{t\geq 0}E[|w_{t+1,i}|^{2}|\mathscr{F}_t]<\infty, {\rm a.s.}, \label{noise}
\ena
then under  \cref{a1} similar results as those of \Cref{l3} and \Cref{l5} can also be obtained, save that the term $``\log r_t(p,q)$" in  \Cref{l3} and \Cref{l5} is replaced by $``\log r_t(p,q)$
$(\log\log r_t(p,q))^{\tau}\ \ \hbox{(for ~some} ~\tau>1)$".
\end{remark}

Now, we present the main results concerning the convergence of the order estimates obtained by \cref{algorithm2}.
\begin{theorem}\label{t1} Under  Assumptions \ref{a1}-\ref{a5}, the order estimate sequence $(p_{t,i},q_{t,i})$ given by (\ref{1}) converges to the true order $(p_0, q_0)$ almost surely, i.e.,
\ban
(p_{t,i},q_{t,i})\xrightarrow[{t\rightarrow\infty}] {} (p_0,q_0), ~~~{\rm a.s.} ~~\hbox{for}~~ i\in\{1,\cdots,n\}.
\ean
\end{theorem}

\begin{proof}
 For  $i\in\{1,\cdots,n\}$, we need to show that the sequence $(p_{t,i},q_{t,i})$ has only one limit point $(p_0,q_0)$.  Let $(p'_i,q'_i)\in M$
be a limit point  of $(p_{t,i},q_{t,i})$, i.e., there is a subsequence  $\{t_k\}$ such that
\bna
(p_{t_k,i},q_{t_k,i})\xrightarrow[{k\rightarrow\infty}] {} (p'_i,q'_i).\label{2}
\ena
In order to prove $(p_{t,i},q_{t,i})\xrightarrow[{t\rightarrow\infty}] {} (p_0,q_0)$, we just need to show the impossibility of the following two situations,

(i)~~$p'_i\geq p_0$, $q'_i\geq q_0$ and $p'_i+q'_i> p_0+q_0$,

(ii)~~$p'_i<p_0 $ or $q'_i<q_0 $.

Note that both $p_{t_k,i}$ and $q_{t_k,i}$ are integers, by (\ref{2}) we have $(p_{t_k,i},q_{t_k,i})\equiv(p'_i,q'_i)$ for sufficiently large $k$.
We first show that the situation (i) will not happen by reduction to absurdity.

Suppose that (i) holds.  By (\ref{dynamics}) and (\ref{w4}), we see that $\sigma_{t_k,i}(p'_i,q'_i,\bm\theta_{t_k,i}(p'_i,q'_i))$ can be calculated by the following equation,
\bna
&&\sigma_{t_k,i}(p'_i,q'_i,\bm\theta_{t_k,i}(p'_i,q'_i))\nonumber\\
&=&\sum^n_{j=1}\sum^{t_k-1}_{l=0}a^{(t_k-l)}_{ij}[y_{l+1,j}
-{\bm\theta^T_{t_k,i}}(p'_i,q'_i)\bm\varphi_{l,j}(p'_i,q'_i)]^2\nonumber\\
&=&\sum^n_{j=1}\sum^{t_k-1}_{l=0}a^{(t_k-l)}_{ij}[ \widetilde{\bm\theta}^T_{t_k,i}(p'_i,q'_i)\bm\varphi_{l,j}(p'_i,q'_i)+w_{l+1,j}]^2\nonumber\\
&=&\widetilde{\bm\theta}^T_{t_k,i}(p'_i,q'_i)
\left(\sum^n_{j=1}\sum^{t_k-1}_{l=0}a^{(t_k-l)}_{ij}\bm\varphi_{l,j}(p'_i,q'_i)\bm\varphi^T_{l,j}(p'_i,q'_i)\right) \widetilde{\bm\theta}_{t_k,i}(p'_i,q'_i)\nonumber\\
&&+2\widetilde{\bm\theta}^T_{t_k,i}(p'_i,q'_i)\left(\sum^n_{j=1}\sum^{t_k-1}_{l=0}a^{(t_k-l)}_{ij}\bm\varphi_{l,j}(p'_i,q'_i)w_{l+1,j}\right)+\sum^n_{j=1}\sum^{t_k-1}_{l=0}a^{(t_k-l)}_{ij}w^2_{l+1,j}.\label{7}
\ena
By \Cref{l3} and \Cref{l5}, we have the following relationship,
\bna
&&\left|\widetilde{\bm\theta}^T_{t_k,i}(p'_i,q'_i)\left(\sum^n_{j=1}\sum^{t_k-1}_{l=0}a^{(t_k-l)}_{ij}\bm\varphi_{l,j}(p'_i,q'_i)w_{l+1,j}\right)\right|\nonumber\\
&\leq& \Bigg\|\widetilde{\bm\theta}^T_{t_k,i}(p'_i,q'_i)\bm P^{-\frac{1}{2}}_{t_k,i}(p'_i,q'_i)\Bigg\|\Bigg\|\bm P^{\frac{1}{2}}_{t_k,i}(p'_i,q'_i)
\left(\sum^n_{j=1}\sum^{t_k-1}_{l=0}a^{(t_k-l)}_{ij}\bm\varphi_{l,j}(p'_i,q'_i)w_{l+1,j}\right)\Bigg\|\nonumber\\
&&=O(\log(r_{t_k}(p'_i,q'_i))=O(\log(r_{t_k}(p^*,q^*)).\label{www10}
\ena
By (\ref{ws4}) and (\ref{w8}),  we have for any $p$ and $q$
\bna
{{\bm P}}^{-1}_{t_k,i}(p,q)=\sum^n_{j=1}\sum^{t_k-1}_{l=0}a^{(t_k-l)}_{ij}\bm\varphi_{l,j}(p,q)\bm\varphi^T_{l,j}(p,q)+\sum^n_{j=1}a^{(t_k)}_{ij}\bm P^{-1}_{0,j}(p,q).\label{www2}
\ena
By this equation and \Cref{l3}, we can easily obtain that
 \bna \label{11}&&\widetilde{\bm\theta}^T_{t_k,i}(p'_i,q'_i)\sum^n_{j=1}\sum^{t_k-1}_{l=0}a^{(t_k-l)}_{ij}\bm\varphi_{l,j}(p,q)\bm\varphi^T_{l,j}(p'_i,q'_i) \widetilde{\bm\theta}_{t_k,i}(p'_i,q'_i)\nonumber\\&=&O(\log r_{t_k}(p'_i,q'_i))=O(\log(r_{t_k}(p^*,q^*)).
 \ena
 Substituting  (\ref{www10}) and (\ref{11}) into (\ref{7}), we see that there exists a positive constant $C_1$ satisfying
\bna
&&\sigma_{t_k,i}(p'_i,q'_i,\bm\theta_{t_k,i}(p'_i,q'_i))-\sum^n_{j=1}\sum^{t_k-1}_{l=0}a^{(t_k-l)}_{ij}w^2_{l+1,j}
\geq-C_1\log(r_{t_k}(p^*,q^*).
\label{13}
\ena

Now, we will consider $\sigma_{t_k,i}(p_0,q_0,\bm\theta_{t_k,i}(p_0,q_0))$. By \cref{wl2}, we have for $p\geq p_0$ and $q\geq q_0$,
\bna
&&\sum^n_{j=1}\sum^{t_k-1}_{l=0}a^{(t_k-l)}_{ij}\bm \varphi_{l,j}(p,q)w_{l+1,j}\nonumber\\&=&\sum^n_{j=1}a^{(t_k)}_{ij}\bm P^{-1}_{0,j}(p,q)\widetilde{\bm \theta}_{0,j}(p,q)-\bm P^{-1}_{t_k,i}(p,q)\widetilde{\bm \theta}_{t_k,i}(p,q).\label{www1}
\ena
By a similar way as that used in (\ref{7}), we obtain
\bna
&&\sigma_{t_k,i}(p_0,q_0,\bm\theta_{t_k,i}(p_0,q_0))-\sum^n_{j=1}\sum^{t_k-1}_{l=0}a^{(t_k-l)}_{ij}w^2_{l+1,j}\nonumber\\
&=&\widetilde{\bm\theta}^T_{t_k,i}(p_0,q_0)
\left(\sum^n_{j=1}\sum^{t_k-1}_{l=0}a^{(t_k-l)}_{ij}\bm\varphi_{l,j}(p_0,q_0)\bm\varphi^T_{l,j}(p_0,q_0)\right) \widetilde{\bm\theta}_{t_k,i}(p_0,q_0)\nonumber\\
&&+2\widetilde{\bm\theta}^T_{t_k,i}(p_0,q_0)\left(\sum^n_{j=1}\sum^{t_k-1}_{l=0}a^{(t_k-l)}_{ij}\bm\varphi_{l,j}
((p_0,q_0)w_{l+1,j}\right).\nonumber
\ena
Combining this with (\ref{www2}) and (\ref{www1}) yields
\bna
&&\sigma_{t_k,i}(p_0,q_0,\bm\theta_{t_k,i}(p_0,q_0))-\sum^n_{j=1}\sum^{t_k-1}_{l=0}a^{(t_k-l)}_{ij}w^2_{l+1,j}\nonumber\\
&=&\widetilde{\bm\theta}^T_{t_k,i}(p_0,q_0){{\bm P}}^{-1}_{t_k,i}(p_0,q_0) \widetilde{\bm\theta}_{t_k,i}(p_0,q_0)\nonumber\\
&&-\widetilde{\bm\theta}^T_{t_k,i}(p_0,q_0)\left(\sum^n_{j=1}a^{(t_k)}_{ij}\bm P^{-1}_{0,j}(p_0,q_0)\right) \widetilde{\bm\theta}_{t_k,i}(p_0,q_0)\nonumber\\
&&+2\widetilde{\bm\theta}^T_{t_k,i}(p_0,q_0)\left(\sum^n_{j=1}a^{(t_k)}_{ij}\bm P^{-1}_{0,j}(p_0,q_0)\widetilde{\bm\theta}_{0,j}(p_0,q_0)-\bm P^{-1}_{t_k,i}(p_0,q_0)\widetilde{\bm\theta}_{t_k,i}(p_0,q_0)\right)\nonumber\\
&\leq&-\widetilde{\bm\theta}^T_{t_k,i}(p_0,q_0)\left(\sum^n_{j=1}a^{(t_k)}_{ij}\bm P^{-1}_{0,j}(p_0,q_0)\right) \widetilde{\bm\theta}_{t_k,i}(p_0,q_0)\nonumber\\
&&+2\widetilde{\bm\theta}^T_{t_k,i}(p_0,q_0)\left(\sum^n_{j=1}a^{(t_k)}_{ij}\bm P^{-1}_{0,j}(p_0,q_0)\widetilde{\bm\theta}_{0,j}(p_0,q_0)\right)\nonumber\\
&&\leq  \left(\sum^n_{j=1}a^{(t_k)}_{ij}\widetilde{\bm\theta}^T_{0,j}(p_0,q_0)\bm P^{-1}_{0,j}(p_0,q_0)\widetilde{\bm\theta}_{0,j}(p_0,q_0)\right)=O(1),\label{12}
\ena
where the last inequality is obtained by
 \bna
 2\bm x^T\bm A\bm y\leq \bm x^T\bm A\bm x+\bm y^T\bm A\bm y  ~~for ~~\bm A\geq 0.\label{www11}
 \ena
From (\ref{13}) and (\ref{12}), we have,
\ban
&&\sigma_{t_k,i}(p'_{i},q'_{i},\bm\theta_{t_k,i}(p'_{i},q'_{i}))-\sigma_{t_k,i}(p_0,q_0,\bm\theta_{t_k,i}(p_0,q_0))\geq -C_1\log r_{t_k}(p^*,q^*)-C_2,
\ean
where $C_2$ is a positive constant.
Note that  $(p_{t_k,i},q_{t_k,i})=\arg\min_{p,q\in M}L_{t_k,i}(p,q)$. By \Cref{a5}, we have
\ban
0&\geq& L_{t_k,i}(p_{t_k,i},q_{t_k,i})- L_{t_k,i}(p_{0},q_{0})=L_{t_k,i}(p'_{i},q'_{i})- L_{t_k,i}(p_{0},q_{0})\\
&=&\sigma_{t_k,i}(p'_i,q'_i,\bm\theta_{t_k,i}(p'_{i},q'_{i}))-\sigma_{t_k,i}(p_0,q_0,\bm\theta_{t_k,i}(p_{0},q_{0}))
+(p'_i+q'_i-p_0-q_0)a_{t_k}\\
&\geq& -C_1\log r_{t_k}(p^*,q^*)-C_2+(p'_i+q'_i-p_0-q_0)a_{t_k}\\
&=& a_{t_k}\left(\frac{-C_1\log r_{t_k}(p^*,q^*)}{a_{t_k}}+(p'_i+q'_i-p_0-q_0)\right)-C_2\rightarrow\infty, \hbox{as}\ k\rightarrow\infty,
\ean
which leads to the contradiction. Thus, the situation (i) will not happen.

In the following, we will show the impossibility of situation (ii) by reduction to absurdity. Suppose that (ii) holds, i.e., $p'_i<p_0 $ or $q'_i<q_0 $.
%For such a case,  {\color{red} based on $\{\bm\varphi_{l,j}(p'_i, q'_i)\}$, from \cref{algorithm1}, we can only obtain the estimates for the partial components of $\bm\theta(p'_i\vee p_0, q'_i\vee q_0)$, i.e., $\bm\theta_{t_k,i}(p'_i, q'_i)\triangleq [b^i_{1,t_k},\cdots, b^i_{{p'_i},t_k},c^i_{1,t_k},\cdots,c^i_{q'_{i},t_k}]^T$.}
 In order to analyze the properties of the estimate error, we introduce the following $(s_i+v_i)$-dimensional vector with $s_i=\max\{p_0, p'_i\}$, $v_i=\max\{q_0, q'_i\}$,
\ban \bm\theta_{t_k,i}(s_i,v_i)=[b^i_{1,t_k},\cdots, b^i_{{s_i},t_k},c^i_{1,t_k},\cdots,c^i_{v_{i},t_k}]^T.\ean
If $p'_i<p_0 $, then $b^i_{m,t_k}\triangleq0$ for $p'_i<m\leq p_0$;  and if $q'_i<q_0 $, then $c^i_{m,t_k}\triangleq0$ for $q'_i<m\leq q_0$.

Denote
$ \widetilde{\bm\theta}_{t_k,i}(s_i,v_i)=\bm\theta(s_i,v_i)-\bm\theta_{t_k,i}(s_i,v_i). $ It is clear that\bna
\| \widetilde{\bm\theta}_{t_k,i}(s_i,v_i)\|^2\geq \min\{|b_{p_0}|^2,|c_{q_0}|^2\}\triangleq \alpha_0>0.\label{17}
\ena
Then by (\ref{2.6}), we have
\ban
&&\sigma_{t_k,i}(p'_i,q'_i,\bm\theta_{t_k,i}(p'_i,q'_i))\\
&&=\sum^n_{j=1}\sum^{t_k-1}_{l=0}a^{(t_k-l)}_{ij}[{\bm\theta^T}(p_0,q_0)\bm\varphi_{l,j}(p_0,q_0)
-{\bm\theta^T_{t_k,i}}(p'_i,q'_i)\bm\varphi_{l,j}(p'_i,q'_i)+w_{l+1,j}]^2.
\ean
Hence combining this equation and the definition $\widetilde{\bm\theta}_{t_k,i}(s_i,v_i)$, we obtain
\bna
&&\sigma_{t_k,i}(p'_i,q'_i,\bm\theta_{t_k,i}(p'_i,q'_i))-\sum^n_{j=1}\sum^{t_k-1}_{l=0}a^{(t_k-l)}_{ij}w^2_{l+1,j}\nonumber\\
&=&\widetilde{\bm\theta}^T_{t_k,i}(s_i,v_i)
\bm P^{-1}_{t_k,i}(s_i,v_i) \widetilde{\bm\theta}_{t_k,i}(s_i,v_i)-\widetilde{\bm\theta}^T_{t_k,i}(s_i,v_i)\left(\sum^n_{j=1}a^{(t)}_{ij}\bm P^{-1}_{0,j}(s_i,v_i)\right)\widetilde{\bm\theta}_{t_k,i}(s_i,v_i)\nonumber\\
&&+2\widetilde{\bm\theta}^T_{t_k,i}(s_i,v_i)\left(\sum^n_{j=1}\sum^{t_k-1}_{l=0}a^{(t_k-l)}_{ij}\bm\varphi_{l,j}(s_i,v_i)w_{l+1,j}\right)\nonumber\\
&&~~~~~\triangleq\widetilde{\bm\theta}^T_{t_k,i}(s_i,v_i)
\bm P^{-1}_{t_k,i}(s_i,v_i) \widetilde{\bm\theta}_{t_k,i}(s_i,v_i) -M_1+M_2.\label{www8}
\ena
By  (\ref{www2}) and  \cref{remark31}, we have for any $p$ and $q$
\bna
\lambda_{\min}(\bm P^{-1}_{t_k,i}(p,q))
\geq a_{\min}\lambda^{p,q}_{\min}(t_k),\label{order4}
\ena
where $a_{\min}=\min_{i,j\in\mathcal{V}}a^{(D_{\mathcal{G}})}_{ij}>0.$
Hence, by (\ref{order4}) and \Cref{l3}, we have for $p\geq p_0$ and $q\geq q_0$,
\bna
\sum^n_{i=1}\|\widetilde{\bm\theta}_{t+1,i}(p,q)\|^2=O\left(\frac{\log r_t(p,q)}{\lambda^{p,q}_{\min}(t)}\right).
\label{order2}
\ena
When $p'_i<p_0$ (so does the case $q'_i<q_0$), we can use (\ref{order2}) in  the first $p'_i$ components of $\widetilde{\bm\theta}_{t_k,i}(s_i,v_i)$. Then by (\ref{mineig}) and \Cref{a5}, we obtain $\|\widetilde{\bm\theta}_{t_k,i}(s_i,v_i)\|=O(1)$, hence we have
\bna
M_1
\leq \lambda_{\max}\left(\sum^n_{j=1}a^{(t)}_{ij}\bm P^{-1}_{0,j}(s_i,v_i)\right)\|\widetilde{\bm\theta}_{t_k,i}(s_i,v_i)\|^2\label{www7}
=O(1).
\ena

In the following, we will analyze $M_2$. Similar to the analysis of (\ref{www10}), by \Cref{l5}, we have
\bna
|M_2|&\leq&\Big\|\widetilde{\bm\theta}^T_{t_k,i}(s_i,v_i)\bm P^{-\frac{1}{2}}_{t_k,i}(s_i,v_i)\Big\|\cdot\Bigg\|\bm P^{\frac{1}{2}}_{t_k,i}(s_i,v_i)
\left(\sum^n_{j=1}\sum^{t_k-1}_{l=0}a^{(t_k-l)}_{ij}\bm\varphi_{l,j}(s_i,v_i)w_{l+1,j}\right)\Bigg\|\nonumber\\
&=&O\left\{\Big\|\widetilde{\bm\theta}^T_{t_k,i}(s_i,v_i)\bm P^{-\frac{1}{2}}_{t_k,i}(s_i,v_i)\Big\|\cdot\log^{\frac{1}{2}}(r_{t_k}(s_i,v_i))\right\}\nonumber\\
&=&O\left\{\Big\|\widetilde{\bm\theta}^T_{t_k,i}(s_i,v_i)\bm P^{-\frac{1}{2}}_{t_k,i}(s_i,v_i)\Big\|\cdot\log^{\frac{1}{2}}(r_{t_k}(p^*,q^*))\right\}.
\label{new8}
\ena
Therefore, by (\ref{www8})-(\ref{new8}), we see that there exist two positive constants $C_3$ and $C_4$ such that
\bna
&&\sigma_{t_k,i}(p'_i,q'_i,\bm\theta_{t_k,i}(p'_i,q'_i))-\sum^n_{j=1}\sum^{t_k-1}_{l=0}a^{(t_k-l)}_{ij}w^2_{l+1,j}\nonumber\\
&\geq&\widetilde{\bm\theta}^T_{t_k,i}(s_i,v_i)
\bm P^{-1}_{t_k,i}(s_i,v_i) \widetilde{\bm\theta}_{t_k,i}(s_i,v_i)-C_3\nonumber\\
&&-C_4\Big\|\widetilde{\bm\theta}^T_{t_k,i}(s_i,v_i)\bm P^{-\frac{1}{2}}_{t_k,i}(s_i,v_i)\Big\|\cdot\log^{\frac{1}{2}}(r_{t_k}( p^*,q^*)).\label{order3}
\ena
By \Cref{a5}, (\ref{17}) and (\ref{order4}), we have
\ban
\Big\|\widetilde{\bm\theta}^T_{t_k,i}(s_i,v_i)\bm P^{-\frac{1}{2}}_{t_k,i}(s_i,v_i)\Big\|\cdot\log^{\frac{1}{2}}(r_{t_k}( p^*,q^*))
=o(\widetilde{\bm\theta}^T_{t_k,i}(s_i,v_i)
\bm P^{-1}_{t_k,i}(s_i,v_i) \widetilde{\bm\theta}_{t_k,i}(s_i,v_i)).
\ean

Furthermore,  by (\ref{17}) and (\ref{order4}), we have
\bna
&&\sigma_{t_k,i}(p'_i,q'_i,\bm\theta_{t_k,i}(p'_i,q'_i))-\sum^n_{j=1}\sum^{t_k-1}_{l=0}a^{(t_k-l)}_{ij}w^2_{l+1,j}\nonumber\\
&=& a_{\min}\alpha_0\lambda^{s_i,v_i}_{\min}(t_k)(1+o(1))-C_3\nonumber\\
&\geq&\frac{a_{\min}\alpha_0\min\{\lambda^{p_0,q^*}_{\min}(t_k),\lambda^{p^*,q_0}_{\min}(t_k)\}}{2}-C_3,\label{18}
\ena
where (\ref{mineig}) is used in the last inequality.

By  (\ref{criterion}), (\ref{12}), (\ref{18}) and \Cref{a5}, for large $k$ and some positive constant $C_5$, we have the following inequality for $i\in\{1,2,\cdots,n\}$,
\ban
0&\geq& L_{t_k,i}(p_{t_k,i},q_{t_k,i})- L_{t_k,i}(p_{0},q_{0})=L_{t_k,i}(p'_{i},q'_{i})- L_{t_k,i}(p_{0},q_{0})\\
&=&\sigma_{t_k,i}(p'_i,q'_i,\bm\theta_{t_k,i}(p'_{i},q'_{i}))-\sigma_{t_k,i}(p_0,q_0,\bm\theta_{t_k,i}(p_{0},q_{0}))
+(p'_i+q'_i-p_0-q_0)a_{t_k}\\
&\geq&\frac{a_{\min}\alpha_0\min\{\lambda^{p_0,q^*}_{\min}(t_k),\lambda^{p^*,q_0}_{\min}(t_k)\}}{2}-C_5+(p'_i+q'_i-p_0-q_0)a_{t_k}\\
&\geq&\frac{a_{\min}\alpha_0\min\{\lambda^{p_0,q^*}_{\min}(t_k),\lambda^{p^*,q_0}_{\min}(t_k)\}}{2}\Big(1+\frac{2(p'_i+q'_i-p_0-q_0)a_{t_k}}
{{a_{\min}\alpha_0\min\{\lambda^{p_0,q^*}_{\min}(t_k),\lambda^{p^*,q_0}_{\min}(t_k)\}}}\Big)-C_5\\
&\geq &\frac{a_{\min}\alpha_0\min\{\lambda^{p_0,q^*}_{\min}(t_k),\lambda^{p^*,q_0}_{\min}(t_k)\}}{4}-C_5\rightarrow\infty,
\ean
which leads to contradiction. The proof of the theorem is complete.
\end{proof}

\begin{remark}
Under  \cref{a1} and the weaker noise condition (\ref{noise}), by  \cref{rrr1}, we can verify that the result of \Cref{t1}
still holds by taking the sequence $\{a_t\}$ in \Cref{a5} to satisfy the following conditions,
\ban
\frac{\log {r_t(p^*,q^*)}(\log\log {r_t(p^*,q^*)})^{\tau}}{a_t}\xrightarrow[{t\rightarrow\infty}] {} 0, ~~
\frac{a_t}{\lambda^{p,q}_{\min}(t)}\xrightarrow[{t\rightarrow\infty}] {} 0,\ \  ~~~{\rm a.s.},
\ean
where $(p,q)\in M^*$.
\end{remark}

\begin{remark} In \Cref{algorithm2}, if the estimates $(p_{t,i}, q_{t,i})$ for the system order $(p_0, q_0)$ are obtained by the following step,
\bna
p_{t,i}={\arg\min}_{1\leq p\leq p^*}L_{t,i}(p,q^*),~~q_{t,i}={\arg\min}_{1\leq q\leq q^*}L_{t,i}(p^*,q),\label{www4}
\ena
then the sequence $(p_{t,i}, q_{t,i})$ can also converge to the true order $(p_0, q_0)$ by a similar argument
as that used in \Cref{t1}. At each time instant $t$, we need to search $p^*\times q^*$ points to find the minimum of the function $L_{t,i}(p,q)$ in (\ref{1}), while in (\ref{www4}) we just need to search at most $p^*+q^*$ points.
\end{remark}

Note that both the estimates $(p_{t,i}, q_{t,i})$  and the true orders $(p_0, q_0)$ are integers,
from \Cref{t1}, we see that there exists a large enough $T$ such that $p_{t,i}=p_0$ and $q_{t,i}=q_0$ for $t\geq T$ . Thus, from (\ref{order2}) and \Cref{a5},
we have the following  consistent estimation of the parameter vector $\bm \theta(p_0,q_0)$.
\begin{theorem}\label{knownparameter}
Under the conditions of \Cref{t1}, for any $i\in\{1,\cdots,n\}$, we have
\ban
\bm \theta_{t,i}(p_{t,i},q_{t,i})\xrightarrow[{t\rightarrow\infty}] {}  \bm \theta(p_0,q_0), ~~~{\rm a.s.}
\ean
where $\bm \theta_{t,i}(p_{t,i},q_{t,i})$ is obtained by \cref{algorithm2}.
\end{theorem}

\section{Case II: The upper bounds of true orders are unknown}\label{section4}
In this section, we consider a general case where the upper bounds of the system orders are unknown. We first give some assumptions on the system signals and the noise.
\begin{assumption}\label{wa2}
For $i\in\{1,\cdots,n\}$, the noise sequence $\{w_{t,i},\mathscr{F}_t\}$ is a martingale difference sequence satisfying
\ban
\sup_{t\geq 0}E[|w_{t+1,i}|^{2}|\mathscr{F}_t]<\infty, ~~\|w_{t,i}\|=O(\eta_i(t)) ~~{\rm a.s.}
\ean
where
$\mathscr{F}_t$  is defined in \Cref{a2} and  $\eta_i(t)$ is a positive, deterministic, nondecreasing function satisfying
\ban
\sup_t \eta_i(e^{t+1})/\eta_i(e^t)<\infty.
\ean
\end{assumption}

In order to simplify the analysis of the estimation error, we need to introduce an assumption on the input and output signals which implies that
the system is not explosive. This assumption is commonly used in the stability analysis of the closed-loop feedback control systems for a single sensor case (see e.g., \cite{wc1}, \cite{wwc3} and \cite{wc5}).
\begin{assumption}\label{wa3}
There exists a positive constant $b$ such that the input and output signals satisfy
\bna
\sum^n_{i=1}\sum^{t-1}_{k=0}(\|y_{k,i}\|^2+\|u_{k,i}\|^2)=O(t^b) ~~ {\rm a.s.}\label{www13}
\ena
\end{assumption}

Similar to \cref{a5} in \cref{www12}, we introduce the following cooperative excitation condition which can be considered
as an extension of the excitation condition used in \cite{wc5} for a single sensor to the distributed order estimation algorithm when the upper bounds of true orders are unknown. This condition can also reflect the joint effect of multiple sensors as illustrated in \cref{joint}.

\begin{assumption}\label{wa4}{\rm (Cooperative Excitation Condition II).}
A sequence $\{\bar a_t\}$ of positive real numbers can be found such that ${\bar{a}_t}\xrightarrow[{t\rightarrow\infty}] {} \infty$ and
\bna
\frac{h_t\log t+[\eta(t)\log\log t]^2}{\bar{a}_t}\xrightarrow[{t\rightarrow\infty}] {} 0, ~~~~~~~~\frac{ {\bar{a}_t}}{\lambda^{0}_{\min}(t)}\xrightarrow[{t\rightarrow\infty}] {} 0,
\ena
hold almost surely, where
\ban
\lambda^{0}_{\min}(t)= \lambda_{\min}\left\{\sum^n_{j=1}\bm P^{-1}_{0,j}(m_0,m_0)+\sum^n_{i=1}\sum^{t-D_{\mathcal{G}}}_{k=0}\bm \varphi^0_{k,i}(\bm \varphi^{0}_{k,i})^T\right\},
\ean
with $\eta(t)\triangleq (\sum^n_{i=1}\eta^2_i(t))^{\frac{1}{2}}$, $\bm \varphi^0_{t,i}=[y_{t,i},\cdots, y_{t-m_0+1,i},u_{t,i},\cdots,u_{t-m_0+1,i},]^T$, $m_0\triangleq\max\{p_0,q_0\}$ and the regression lag
$h_t$ is chosen as $h_t=O((\log t)^{\alpha})(\alpha>1)$, and $\log t=o(h_t)$.
\end{assumption}

We are now to construct the algorithm to estimate both the system orders and parameters in a distributed way when the upper bounds of orders are unknown.
For estimating the unknown orders $(p_0, q_0)$, we introduce the following local information criterion $\bar {L}_{t,i}(p,q)$ for the sensor $i$,
\bna
\bar {L}_{t,i}(p,q)=\sigma_{t,i}(p,q,\bm\theta_{t,i}(p,q))+(p+q)\bar{a}_t,\label{wcriterion}
\ena
where $\sigma_{t,i}(p,q,\bm\beta(p,q))$ is recursively defined in (\ref{3}).

By minimizing the local information criterion (\ref{wcriterion}) and using \cref{algorithm1}, we obtain the following distributed algorithm.
\begin{algorithm}[!ht]
\caption{ }\label{algorithm3}
 For any given sensor $i\in\{1,\cdots,n\}$,  the distributed algorithm for the estimation of the system orders and the parameters  is  defined at the time instant $t\geq1$ as follows.

$\mathbf{Step\ 1:}$ For any $0\leq s\leq [\log t]$, based on $\{\bm \varphi_{k,j}(s,s), y_{k+1,j}\}^{t-1}_{k=1}$,
$j\in N_i$,
 the estimate  $\bm\theta_{t,i}(s,s)$  can  obtained by using \Cref{algorithm1}, where the orders $(p, q)$ are replaced by  $(s,s)$ ($0\leq s\leq [\log t]$).

$\mathbf{Step\ 2:}$ (Order estimation)  With the estimates $\{\bm\theta_{t,i}(s,s)\}^{[\log t]}_{s=0}$ obtained by  Step 1,

\hskip 2.5cm take $\hat m_{t,i}$ by minimizing $\bar {L}_{t,i}(s,s)$ for $0\leq s\leq [\log t]$;

\hskip 2.5cm take $\hat p_{t,i}$ by minimizing  $\bar {L}_{t,i}(p,\hat m_{t,i})$ for $0\leq p\leq \hat m_{t,i}$;

\hskip 2.5cm take $\hat q_{t,i}$ by minimizing  $\bar {L}_{t,i}(\hat p_{t,i},q)$ for $0\leq q\leq \hat m_{t,i}$.

$\mathbf{Step\ 3:}$ (Parameter estimation)
The  estimate $\bm\theta_{t,i}(\hat p_{t,i},\hat q_{t,i})$ for the unknown parameter vector $\bm\theta(p_0,q_0)$ is  obtained by using \Cref{algorithm1}, where the orders $(p, q)$ are replaced by the
estimates $(\hat p_{t,i},\hat q_{t,i})$ obtained by  Step 2.

$\mathbf{Output:}$ $\hat p_{t,i},\hat q_{t,i}$  and $\bm\theta_{t,i}(\hat p_{t,i},\hat q_{t,i})$.
\end{algorithm}

\begin{remark}
In Step 2 of the above \cref{algorithm3}, we first estimate the maximum value $m_0$ of true orders whose upper bound is characterized by the function $\log t$. Then the true orders $p_0, q_0$ are obtained by searching among at most $2\hat m_{t,i}$ points at each time instant $t$.
\end{remark}

In the following, we will provide the consistency analysis of \Cref{algorithm3} when the upper bounds of orders are unknown, in which a crucial step is to prove that for any $i$, $\hat m_{t,i} \rightarrow m_0$ as $t\rightarrow\infty$. Then  by  the order estimation procedure in \Cref{algorithm3}, the convergence of the estimates for the system orders and parameters can be obtained by a similar analysis as those in \cref{www12}. To this end, we need to introduce the following double array  martingale estimation lemma to deal with the noise effect in the form of $\max_{1\leq m\leq h_t}\left\|\sum^t_{k=1}f_k(m)w_{k+1}\right\|$.

\begin{lemma}{\rm \cite{wc5}}\label{wlb}
Let $\{v_t,\mathscr{F}_t\}$ be an $s'$-dimensional martingale difference sequence satisfying
$\|v_t\|=o(\rho(t))~~{\rm a.s.},$, $\sup_t E(\|v_{t+1}\|^2|\mathscr{F}_t)<\infty$ {\rm a.s.},
 where the properties of $\rho(t)$ is described as same as $\eta_i(t)$ in \Cref{wa2}. Assume that $f_t(m), t, m=1,2,...,$ is an $\mathscr{F}_t$-measurable, $r'\times s'$-dimensional random matrix satisfying $\|f_t(m)\|\leq C<\infty$ {\rm a.s.} for all t, m and some deterministic constant C. Then for $h_t=O([\log t]^{\alpha})$ ~$(\alpha>1)$, the following property holds as $t\rightarrow\infty$,
 \ban
 \max_{1\leq m\leq h_t}\max_{1\leq j\leq t}\left\|\sum^j_{k=1}f_k(m)v_{k+1}\right\|=O\left(\max_{1\leq m\leq h_t}\sum^t_{k=1}\|f_k(m)\|^2\right)+o(\rho(t)\log\log t),~~{\rm a.s.}
 \ean
\end{lemma}

In order to simplify the expression of the following lemmas and theorems, we write $(l)$ for $(l,l)$ in $\bm \theta_{t,i}, \bm \varphi_{t,i}$ and $\bm P_{t,i}$ when $p=q=l$.

\begin{lemma} \label{wl6}
Let $V_t(l)=\widetilde{\bm \Theta}^T_t(l)\bm P^{-1}_t(l)\widetilde{\bm \Theta}_t(l)$. Then under \cref{a1} and Assumptions \ref{wa2}-\ref{wa4}, we have
\ban
\max_{ m_0\leq l\leq h_t}V_{t+1}(l)=O(h_t\log t)+o(\eta^2(t)\log\log t),
\ean
where $h_t$ and $\eta(t)$ are defined in \Cref{wa4}.
\end{lemma}
\begin{proof}
By the proof of Lemma 4.4 in \cite{wc2}, we have for $l\geq m_0$
\bna
V_{t+1}(l)
= O(1)+\sum^t_{k=0}\bm W^T_{k+1}\bm d_k(l)\bm \Phi^T_k(l)\bm P_k(l)\bm \Phi_k(l)\bm W_{k+1}.\label{we5}
\ena

By (\ref{www2}) and \Cref{wla}, we have
\bna
&&\max_{1\leq l\leq h_t}\sum^t_{k=0}\lambda_{\max}\{\bm d_k(l)\bm \Phi^T_k(l)\bm P_k(l)\bm \Phi_k(l)\}\nonumber\\
&\leq&\max_{1\leq l\leq h_t}[\log det(\bm P_{t+1}^{-1}(l))-\log det(\bm P_{0}^{-1}(l))]\nonumber\\
&=&O\bigg\{\max_{1\leq l\leq h_t}\Big(l\cdot\log \Big(\lambda_{\max}{\bm P_0^{-1}(l)}+\sum^n_{j=1}\sum^t_{k=0}\|\bm \varphi_{k,j}(l)\|^2\Big)\Big)\bigg\}=O(h_t\log t),\label{we3}
\ena
where \Cref{wa3} is used in the last equation. By \Cref{wa2}, \Cref{wla} and \Cref{wlb}, we have
\ban
&&\max_{1\leq l\leq h_t}\sum^t_{k=0}\lambda_{\max}\{\bm d_k(l)\bm \Phi^T_k(l)\bm P_k(l)\bm \Phi_k(l)\}(\|\bm W_{k+1}\|^2-E(\|\bm W_{k+1}\|^2|\mathscr{F}_k))\\
&=&o(\eta^2(t)\log\log t)+O\left(\max_{1\leq l\leq h_t}\sum^t_{k=0}\lambda_{\max}\{\bm d_k(l)\bm \Phi^T_k(l)\bm P_k(l)\bm \Phi_k(l)\}\right).
\ean
Hence by (\ref{we3}) and \Cref{wa2},  we have
\bna
&&\max_{1\leq l\leq h_t}\sum^t_{k=0}\bm W^T_{k+1}\bm d_k(l)\bm \Phi^T_k(l)\bm P_k(l)\bm \Phi_k(l)\bm W_{k+1}\nonumber\\
&\leq&\max_{1\leq l\leq h_t}\sum^t_{k=0}\lambda_{\max}\{\bm d_k(l)\bm \Phi^T_k(l)\bm P_k(l)\bm \Phi_k(l)\}\|\bm W_{k+1}\|^2\nonumber\\
&\leq&\max_{1\leq l\leq h_t}\sum^t_{k=0}\lambda_{\max}\{\bm d_k(l)\bm \Phi^T_k(l)\bm P_k(l)\bm \Phi_k(l)\}(\|\bm W_{k+1}\|^2-E(\|\bm W_{k+1}\|^2|\mathscr{F}_k))\nonumber\\
&&+\max_{1\leq l\leq h_t}\sum^t_{k=0}\lambda_{\max}\{\bm d_k(l)\bm \Phi^T_k(l)\bm P_k(l)\bm \Phi_k(l)\}E(\|\bm W_{k+1}\|^2|\mathscr{F}_k))\nonumber\\
&=&o(\eta^2(t)\log\log t)+O(h_t\log t).\label{we6}
\ena
Substituting (\ref{we6}) into (\ref{we5}) yields the result of the lemma.
\end{proof}

\begin{lemma}\label{wl11}
 Under \Cref{a1}, Assumptions \ref{wa2}-\ref{wa4}, for any $i\in\{1,...,n\}$, we have
\ban
&&\max_{1\leq l\leq h_t}\left\{\bm S^T_{t,i}(l)\bm P_{t,i}(l)\bm S_{ t,i }(l)\right\}
= O(h_t\log t)+o(\{\eta(t)\log\log t\}^2),
\ean
where $\bm S_{t,i}(l)=\left(\sum^n_{j=1}\sum^{t-1}_{k=0}a^{(t-k)}_{ij}\bm\varphi_{k,j}(l)w_{k+1,j}\right)$, and $h_t$ and $\eta(t)$ are defined in \Cref{wa4}.
\end{lemma}

The above lemma can be derived by following the proof of \cref{l5}, and we omit it here.
%By (\ref{we6}), we have
%\bna
%&&\max_{1\leq l\leq h_t}\sum^{t-1}_{k=0}\bm W^T_{k+1}\bm \Phi^T_k(l)\bm P_k(l)\bm \Phi_k(l)\bm d_k(l)\bm W_{k+1}\nonumber\\
%&=&\max_{1\leq l\leq h_t}\sum^{t-1}_{k=0}\bm W^T_{k+1}\bm d_k(l)\bm \Phi^T_k(l)\bm P_k(l)\bm \Phi_k(l)\bm W_{k+1}\nonumber\\
%&=&o(\eta^2(t)\log\log t)+O(h_t\log t).\label{wnew7}
%\ena
%Hence (\ref{wnew5}),(\ref{wnew6})  and (\ref{wnew7}), we have
%\ban
%\max_{1\leq l\leq h_t}\bm S_t(l)\bm P_t(l)\bm S_t(l)=o(\eta^2(t)\log\log t)+O(h_t\log t).
%\ean
%which completes the proof.

The following theorem will establish the convergence of the estimates $\hat m_{t,i}, \hat p_{t,i},\hat q_{t,i}$ and $\bm\theta_{t,i}(\hat p_{t,i},\hat q_{t,i})$  given by \cref{algorithm3} to the true values.
\begin{theorem}\label{unknown}
Under \Cref{a1}, Assumptions \ref{wa2}-\ref{wa4}, we have  for any $i\in\{1,\cdots,n\}$,
\bna
&&\hat m_{t,i}\xrightarrow[{t\rightarrow\infty}] {} m_0 ~~{\rm a.s.}\label{wd1}\\
&&(\hat p_{t,i},\hat q_{t,i})\xrightarrow[{t\rightarrow\infty}] {} (p_0,q_0)~~{\rm a.s.}\label{wd2}\\
&&\bm\theta_{t,i}(\hat p_{t,i},\hat q_{t,i})\xrightarrow[{t\rightarrow\infty}] {}\bm\theta(p_0,q_0)~~{\rm a.s.}\label{wd3}
\ena
\end{theorem}
\begin{proof}
We first show that $\lim\sup_{t\rightarrow\infty}\hat m_i(t)\leq m_0$ ~~a.s.

For $p>p_0$, $q>q_0$, set
\ban
&&\bm \theta(p,q)=[b_1,\cdots,b_p,c_1,\cdots,c_q]^T,\\
&&\widetilde{\bm\theta}_{t,i}(p,q)=\bm \theta(p,q)-\bm\theta_{t,i}(p,q),
\ean
 where $b_p=0, p>p_0$, $c_q=0, q>q_0$, $\bm\theta_{t,i}(p,q)$ is obtained by \cref{algorithm1}.

Then by (\ref{w4}), for $l\geq m_0$, we have
\bna
\sigma_{t,i}(l,l,\bm\theta_{t,i}(l,l))
&=&\sum^n_{j=1}\sum^{t-1}_{k=0}a^{(t-k)}_{ij}[ \widetilde{\bm\theta}^T_{t,i}(l)\bm\varphi_{k,j}(l)+w_{k+1,j}]^2\nonumber\\
&=&\widetilde{\bm\theta}^T_{t,i}(l)
\left(\sum^n_{j=1}\sum^{t-1}_{k=0}a^{(t-k)}_{ij}\bm\varphi_{k,j}(l)\bm\varphi^T_{k,j}(l)\right) \widetilde{\bm\theta}_{t,i}(l)\nonumber\\
&&+2\widetilde{\bm\theta}^T_{t,i}(l)\left(\sum^n_{j=1}\sum^{t-1}_{k=0}a^{(t-k)}_{ij}
\bm\varphi_{k,j}(l)w_{k+1,j}\right)+\sum^n_{j=1}\sum^{t-1}_{k=0}a^{(t-k)}_{ij}w^2_{k+1,j}\nonumber\\
&\triangleq& I_1+I_2+\sum^n_{j=1}\sum^{t-1}_{k=0}a^{(t-k)}_{ij}w^2_{k+1,j}.\label{we20}
\ena
In the following, we estimate $ I_1$, $I_2$ separately.

On the one hand, by  (\ref{www2}) (\ref{www1}) and (\ref{www11}), we have
\bna
 I_1+I_2&=&\widetilde{\bm\theta}^T_{t,i}(l)\left({{\bm P}}^{-1}_{t,i}(l)-\sum^n_{j=1}a^{(t)}_{ij}\bm P^{-1}_{0,j}(l)\right)\widetilde{\bm\theta}_{t,i}(l)\nonumber\\
&&+2\widetilde{\bm\theta}^T_{t,i}(l)\left(\sum^n_{j=1}a^{(t)}_{ij}\bm P^{-1}_{0,j}(l)\widetilde{\bm \theta}_{0,j}(l)-\bm P^{-1}_{t,i}(l)\widetilde{\bm \theta}_{t,i}(l)\right)\nonumber\\
&=&-\widetilde{\bm\theta}^T_{t,i}(l){{\bm P}}^{-1}_{t,i}(l)\widetilde{\bm\theta}_{t,i}(l)+2\widetilde{\bm\theta}^T_{t,i}(l)\left(\sum^n_{j=1}a^{(t)}_{ij}\bm P^{-1}_{0,j}(l)\widetilde{\bm \theta}_{0,j}(l)\right)\nonumber\\
&&-\widetilde{\bm\theta}^T_{t,i}(l)\left(\sum^n_{j=1}a^{(t)}_{ij}\bm P^{-1}_{0,j}(l)\right)\widetilde{\bm\theta}_{t,i}(l)
\nonumber\\
&\leq&-\widetilde{\bm\theta}^T_{t,i}(l){{\bm P}}^{-1}_{t,i}(l)\widetilde{\bm\theta}_{t,i}(l)+\sum^n_{j=1}\left(a^{(t)}_{ij}\widetilde{\bm\theta}^T_{0,j}(l)\bm P^{-1}_{0,j}(l)\widetilde{\bm\theta}_{0,j}(l)\right).\label{we100}
\ena
Hence by (\ref{we20}) and (\ref{we100}), we have for $l\geq m_0$,
\bna
\sigma_{t,i}(l,l,\bm\theta_{t,i}(l,l))-\sum^n_{j=1}\sum^{t-1}_{k=0}a^{(t-k)}_{ij}w^2_{k+1,j}\leq-\widetilde{\bm\theta}^T_{t,i}(l){{\bm P}}^{-1}_{t,i}(l)\widetilde{\bm\theta}_{t,i}(l)+O(1).\label{we101}
\ena

On the other hand, by \Cref{wl11}, we have
\ban
|I_2|
&\leq&2\Bigg\|\widetilde{\bm\theta}^T_{t,i}(l)\bm P^{-\frac{1}{2}}_{t,i}(l)\Bigg\|
\cdot\Bigg\|\bm P^{\frac{1}{2}}_{t,i}(l)
\left(\sum^n_{j=1}\sum^{t-1}_{k=0}a^{(t-k)}_{ij}\bm\varphi_{k,j}(l)w_{k+1,j}\right)\Bigg\|\nonumber\\
&\leq&O\left\{\Big\|\widetilde{\bm\theta}^T_{t,i}(l)\bm P^{-\frac{1}{2}}_{t,i}(l)\Big\|\cdot \{o([\eta(t)\log\log t]^2)+O(h_t\log t)\}^{\frac{1}{2}}\right\}.
%\label{wnew8}
\ean
Then for $l\geq m_0$ and sufficiently large $t$,  by (\ref{www2}), (\ref{we20}), and \Cref{wa4} ,
we have for some positive constant $C_6$,
\ban
&&\sigma_{t,i}(l,l,\bm\theta_{t,i}(l,l))-\sum^n_{j=1}\sum^{t-1}_{k=0}a^{(t-k)}_{ij}w^2_{k+1,j}\\
&\geq& \frac{1}{2}\widetilde{\bm\theta}^T_{t,i}(l){{\bm P}}^{-1}_{t,i}(l)\widetilde{\bm\theta}_{t,i}(l)\\
&&-C_6\left\{\Big\|\widetilde{\bm\theta}^T_{t,i}(l)\bm P^{-\frac{1}{2}}_{t,i}(l)\Big\|\cdot \{o([\eta(t)\log\log t]^2)+h_t\log t\}^{\frac{1}{2}}\right\}.
\ean
Hence by (\ref{wcriterion}) and \Cref{wl6}, we have for sufficiently large $t$,
\ban
&&\max_{m_0<l\leq\log t}\{\bar {L}_{t,i}(m_0,m_0)-\bar {L}_{t,i}(l,l)\}\\
&=&\max_{m_0<l\leq\log t}\Big\{\sigma_{t,i}(m_0,m_0,\bm\theta_{t,i}(m_0,m_0))-\sum^n_{j=1}\sum^{t-1}_{k=0}a^{(t-k)}_{ij}w^2_{k+1,j}\\
&&-\sigma_{t,i}(l,l,\bm\theta_{t,i}(l,l))+\sum^n_{j=1}\sum^{t-1}_{k=0}a^{(t-k)}_{ij}w^2_{k+1,j}-2(l-m_0)\bar{a}_t\Big\}\\
&\leq&\max_{m_0<l\leq\log t}O\left\{\Big\|\widetilde{\bm\theta}^T_{t,i}(l)\bm P^{-\frac{1}{2}}_{t,i}(l)\Big\|\{o([\eta(t)\log\log t]^2)+O(h_t\log t)\}^{\frac{1}{2}}\right\}+O(1)-2\bar{a}_t\\
&=&o([\eta(t)\log\log t]^2)+O(h_t\log t)+O(1)-2\bar{a}_t<0.
\ean
By the above equation, we have
\ban
\bar {L}_{t,i}(m_0,m_0)<\min_{m_0<l\leq \log t}\bar {L}_{t,i}(l,l),
\ean
which implies that $\lim\sup_{t\rightarrow\infty}\hat m_{t,i}\leq m_0$.

We now show that $\lim\inf_{t\rightarrow\infty}\hat m_{t,i}\geq m_0$ holds almost surely.
For any $l\leq m_0$,
 let us write $\bm \theta_{t,i}(l)$ given by \Cref{algorithm1} in its component form
\ban
\bm \theta_{t,i}(l)=[b^i_{1,t},\cdots, b^i_{l,t},c^i_{1,t},\cdots,c^i_{l,t}]^T\in\mathbb{R}^{2l}.
\ean
In order to avoid confusion, for any $l\leq m_0$, we denote the following $m_0$-dimensional vector,
\ban
\bm \theta_{t,i}(m_0)=[b^i_{1,t},\cdots, b^i_{m_0,t},c^i_{1,t},\cdots,c^i_{m_0,t}]^T\in\mathbb{R}^{2m_0},
\ean
where $ b^i_{j,t}=0, c^i_{j,t}=0$ for $m_0>j> l$.

For any $l\leq m_0$, we have
\ban
&&y_{k+1,j}-{\bm\theta}_{t,i}^T(l)\bm\varphi_{k,j}(l)=y_{k+1,j}-{\bm\theta}_{t,i}^{ T}(m_0)\bm\varphi_{k,j}(m_0)\\
&=&y_{k+1,j}-{\bm\theta}^T(m_0)\bm\varphi_{k,j}(m_0)+[{\bm\theta}^T(m_0)-{\bm\theta}_{t,i}^{ T}(m_0)]\bm\varphi_{k,j}(m_0)\\
&=&w_{k+1,j}+\widetilde{\bm \theta}^T_{t,i}(m_0)\bm\varphi_{k,j}(m_0),
\ean
where $\widetilde{\bm \theta}_{t,i}(m_0)={\bm\theta}(m_0)-{\bm\theta}_{t,i}(m_0)\in \mathbb{R}^{2m_0}$.

Hence by (\ref{w4}), we have for any $l\leq m_0$
\ban
\sigma_{t,i}(l,l,\bm\theta_{t,i}(l,l))
=\sum^n_{j=1}\sum^{t-1}_{k=0}a^{(t-k)}_{ij}(w_{k+1,j}+\widetilde{\bm \theta}^T_{t,i}(m_0)\bm\varphi_{k,j}(m_0))^2.
\ean

Thus,  we have for any $l\leq m_0$
\bna
&&\sigma_{t,i}(l,l,\bm\theta_{t,i}(l,l))
-\sum^n_{j=1}\sum^{t-1}_{k=0}a^{(t-k)}_{ij}w^2_{k+1,j}\nonumber\\
&=&\widetilde{\bm\theta}^T_{t,i}(m_0)
\left(\sum^n_{j=1}\sum^{t-1}_{k=0}a^{(t-k)}_{ij}\bm\varphi_{k,j}(m_0)\bm\varphi^T_{k,j}(m_0)\right) \widetilde{\bm\theta}_{t,i}(m_0)\nonumber\\
&&+2\widetilde{\bm\theta}^T_{t,i}(m_0)\left(\sum^n_{j=1}\sum^{t-1}_{k=0}a^{(t-k)}_{ij}\bm\varphi_{k,j}(m_0)w_{k+1,j}\right)
\triangleq J_1+J_2.\label{we9}
\ena
In the following, we estimate $J_1, J_2$.

For $l<m_0$, by the definition of $\widetilde{\bm \theta}_{t,i}(m_0)$, we have
\bna
\| \widetilde{\bm\theta}_{t,i}(m_0)\|^2\geq \min\{|b_{p_0}|^2,|c_{q_0}|^2\}= \alpha_0>0.\label{we8}
\ena
Then by (\ref{www2}), we have
\bna
\widetilde{\bm\theta}^T_{t,i}(m_0)
\bm P^{-1}_{t,i}(m_0)\widetilde{\bm\theta}_{t,i}(m_0)\geq a_{\min}\lambda^0_{\min}(t)\alpha_0.\label{ww2}
\ena
Moreover,    by  (\ref{mineig}), \cref{wa4} and \cref{wl6}, we have
\ban
&&\widetilde{\bm\theta}^T_{t,i}(m_0)\left(\sum^n_{j=1}a^{(t)}_{ij}\bm P^{-1}_{0,j}(m_0)\right)\widetilde{\bm\theta}_{t,i}(m_0)\\
&\leq& \lambda_{\max}\left(\sum^n_{j=1}a^{(t)}_{ij}\bm P^{-1}_{0,j}(m_0)\right)\|\widetilde{\bm\theta}_{t,i}(m_0)\|^2
=O(1).
\ean
Then for $l<m_0$, by (\ref{www2}), we obtain for some positive constant $C_7$,
\bna
J_1&=&\widetilde{\bm\theta}^T_{t,i}(m_0)
\bm P^{-1}_{t,i}(m_0)\widetilde{\bm\theta}_{t,i}(m_0)-\widetilde{\bm\theta}^T_{t,i}(m_0)
\left(\sum^n_{j=1}a^{(t)}_{ij}\bm P^{-1}_{0,j}(m_0)\right)\widetilde{\bm\theta}_{t,i}(m_0)\nonumber\\
&\geq&\widetilde{\bm\theta}^T_{t,i}(m_0)
\bm P^{-1}_{t,i}(m_0)\widetilde{\bm\theta}_{t,i}(m_0)-C_7.\label{we7}
\ena
By \Cref{wl11}, we have
\bna
|J_2|
&\leq&2\Bigg\|\widetilde{\bm\theta}^T_{t,i}(m_0)\bm P^{-\frac{1}{2}}_{t,i}(m_0)\Bigg\|
\cdot\Bigg\|\bm P^{\frac{1}{2}}_{t,i}(m_0)
\left(\sum^n_{j=1}\sum^{t-1}_{k=0}a^{(t-k)}_{ij}\bm\varphi_{k,j}(m_0)w_{k+1,j}\right)\Bigg\|\nonumber\\
&\leq&O\left\{\Big\|\widetilde{\bm\theta}^T_{t,i}(m_0)\bm P^{-\frac{1}{2}}_{t,i}(m_0)\Big\|\cdot \{o([\eta(t)\log\log t]^2)+O(h_t\log t)\}^{\frac{1}{2}}\right\}.
\label{wnew8}
\ena
Then by (\ref{ww2})-(\ref{wnew8}) and \cref{wa4}, we have for large $t$
\ban
J_1+J_2&\geq& \widetilde{\bm\theta}^T_{t,i}(m_0)
\bm P^{-1}_{t,i}(m_0)\widetilde{\bm\theta}_{t,i}(m_0)-C_8\Big\{\Big\|\widetilde{\bm\theta}^T_{t,i}(m_0)\bm P^{-\frac{1}{2}}_{t,i}(m_0)\Big\|\\
&&\cdot \{o([\eta(t)\log\log t]^2)+O(h_t\log t)\}^{\frac{1}{2}}\Big\} -C_7\\
&\geq& a_{\min}\alpha_0 \lambda^{0}_{\min}(t)(1+o(1)) ~{\rm a.s.},
\ean
where $C_8$ is a positive constant.

Hence by (\ref{we9}), we have for any  $l<m_0$,
\bna
\sigma_{t,i}(l,l,\bm\theta_{t,i}(l,l))\geq a_{\min} \alpha_0 \lambda^{0}_{\min}(t)(1+o(1))+\sum^n_{j=1}\sum^{t-1}_{k=0}a^{(t-k)}_{ij}w^2_{k+1,j}.\label{we13}
\ena
Note that  when $l=m_0$, by (\ref{we101}) and \Cref{wl6}, we have
\bna
&&\sigma_{t,i}(m_0,m_0,\bm\theta_{t,i}(m_0,m_0))\nonumber\\
&\leq&-\widetilde{\bm\theta}^T_{t,i}(m_0){{\bm P}}^{-1}_{t,i}(m_0)\widetilde{\bm\theta}_{t,i}(m_0)+\sum^n_{j=1}\sum^{t-1}_{k=0}a^{(t-k)}_{ij}w^2_{k+1,j}+O(1)\nonumber\\
& \leq& O(h_t\log t)+o([\eta(t)\log\log t]^2)+\sum^n_{j=1}\sum^{t-1}_{k=0}a^{(t-k)}_{ij}w^2_{k+1,j}+O(1).\label{we12}
\ena
For any  $l<m_0$, by (\ref{we13})-(\ref{we12}) and  \cref{wa4}, we have
\ban
&&\bar {L}_{t,i}(l,l)-\bar {L}_{t,i}(m_0,m_0)\\
&=&\sigma_{t,i}(l,l,\bm\theta_{t,i}(l,l))-\sigma_{t,i}(m_0,m_0,\bm\theta_{t,i}(m_0,m_0))
+2(l-m_0)\bar{a}_t\\
&\geq& a_{\min}\alpha_0 \lambda^{0}_{\min}(t)(1+o(1))-C_9((h_t\log t)+o([\eta(t)\log\log t]^2))+C_{10}-C_{11}\bar{a}_t\\
&=& a_{\min}\lambda^{0}_{\min}(t)(\alpha_0 +o(1))>0 \ \hbox{as}\ t\rightarrow \infty,
\ean
where $C_9,C_{10},C_{11}$ are positive constants.
Hence we have
\ban
\bar {L}_{t,i}(m_0,m_0)<\min_{1\leq l< m_0}\bar {L}_{t,i}(l,l),
\ean
which implies that $\lim\inf_{t\rightarrow\infty}\hat m_{t,i}\geq m_0~~ {\rm a.s.}$
Thus the first assertion (\ref{wd1}) has been proved.

By $\hat m_{t,i}\xrightarrow[{t\rightarrow\infty}] {} m_0$, the proof of (\ref{wd2}) can be carried out by a similar argument as that used in \cref{www12}.

Note that both the estimates $(p_{t,i}, q_{t,i})$  and the true orders $(p_0, q_0)$ are integers, we see that there exists a large enough $T$ such that $p_{t,i}=p_0$ and $q_{t,i}=q_0$ for $t\geq T$. By the proof of \Cref{wl6}, we have
\ban
V_t(p_0,q_0)=O(h_t\log t)+o(\eta^2(t)\log\log t).
\ean
Therefore,
\ban
\|\bm\theta_{t,i}(p_0,q_0)-\bm\theta(p_0,q_0)\|^2=\frac{O(h_t\log t)+o(\eta^2(t)\log\log t)}{\lambda^{0}_{\min}(t)}.
\ean
The convergence of the parameters can be obtained by \cref{wa4}. This completes the proof of the theorem.
\end{proof}

\begin{remark}
From \cref{unknown} (also \cref{t1} and \cref{knownparameter}), we see that the convergence of the estimates for both the system orders and parameters are derived without using the independency or stationarity assumptions on the regression vectors, which
makes it possible to apply our distributed algorithms to practical
feedback systems.
\end{remark}

\section{Conclusion}\label{section5}

In this paper, we proposed distributed algorithms to simultaneously estimate both the unknown system orders and parameters  by minimizing the LIC and using the distributed LS algorithm.  For the case where the  upper bounds of true orders are known, we show that the estimates of the parameters and the orders can converge to the true values under the cooperative excitation condition introduced in this paper. We note that the convergence results are obtained without using the independency and stationarity assumptions of regression vectors as commonly used in most existing literatures. Moreover, for the case where the upper bounds of true orders are unknown, we constructed similar distributed algorithm to estimate both the parameters and the orders by introducing a time-varying regression lag, and obtained the strong consistency of the distributed algorithm. The cooperative excitation condition can reveal the joint effect of multiple sensors. Many interesting
problems deserve to be further investigated, for example, the distributed order estimation problem of the autoregressive moving average model with exogenous inputs (ARMAX), the recursive distributed algorithm for the order estimation problem.

\bibliographystyle{siamplain}
\bibliography{references}

\end{document}